\documentclass[twocolumn, twoside]{IEEEtran}

\usepackage{mathrsfs}
\usepackage{subfigure}
\usepackage[dvips]{graphicx}
\usepackage[tbtags]{amsmath}
\usepackage{amsfonts}
\usepackage{latexsym}
\usepackage{amssymb}
\usepackage{amsxtra}
\usepackage{bigstrut}
\usepackage{multirow}
\usepackage{array}
\usepackage{tabularx}
\usepackage{xcolor}
\usepackage{algpseudocode}
\usepackage{algorithm}
\usepackage{cite}
\usepackage{soul}

\ifCLASSOPTIONonecolumn 
  \renewcommand{\baselinestretch}{1.38}
  
  \newcommand \SubFigWidth {.47\columnwidth}
\else 
  
  \newcommand \SubFigWidth {0.95\columnwidth}
\fi

\DeclareMathAlphabet{\mathpzc}{OT1}{pzc}{m}{it}
\DeclareMathAlphabet{\mathcalligra}{T1}{calligra}{m}{n}

\newcommand{\norm}[1]{\left\Vert#1\right\Vert}

\newtheorem{thm}{Theorem}

\newtheorem{rmk}{\bf Remark} 

\begin{document}
\title{Secret Key Agreement with Large Antenna Arrays under the Pilot Contamination Attack}

%
%
%
\author{Sanghun~Im,~\IEEEmembership{Student Member,~IEEE,}
        Hyoungsuk~Jeon,~\IEEEmembership{Member,~IEEE,}
        Jinho~Choi,~\IEEEmembership{Senior Member,~IEEE,}
        and~Jeongseok~Ha,~\IEEEmembership{Member,~IEEE}
\IEEEcompsocitemizethanks{
\IEEEcompsocthanksitem S. Im, H. Jeon, and J. Ha are with the Department of Electrical Engineering, Korea Advanced Institute of Science and Technology, Daejeon, Korea (e-mail: sh.im@kaist.ac.kr, h.jeon@kaist.ac.kr, jsha@kaist.edu).
\IEEEcompsocthanksitem J. Choi is with School of Information and Communications, Gwangju Institute of Science and Technology (GIST), Gwangju, Korea (e-mail: jchoi0114@gist.ac.kr).}}

\maketitle

\begin{abstract}
We present a secret key agreement (SKA) protocol for a multi-user 
time-division duplex system where a base-station (BS) with 
a large antenna array (LAA) shares secret keys with users 
in the presence of non-colluding eavesdroppers. 
In the system, when the BS transmits random sequences 
to legitimate users for sharing common randomness, the eavesdroppers 
can attempt the pilot contamination attack (PCA) in which each 
of eavesdroppers transmits its target 
user's training sequence in hopes of 
acquiring possible information leak 
by steering beam towards the eavesdropper. 
We show that there exists a crucial
complementary relation between the received signal 
strengths at the eavesdropper and its target user.
This relation tells us that
the eavesdropper inevitably leaves a trace that
enables us to devise a 
way of measuring the amount of information leakage to 
the eavesdropper even if PCA parameters are unknown. 
To this end, we derive an estimator for the channel gain 
from the BS to the eavesdropper and propose a rate-adaptation 
scheme for adjusting the length of secret key under the PCA. 
Extensive analysis 
and evaluations are carried out under various setups, which 
show that the proposed scheme adequately takes advantage of 
the LAA to establish the secret keys under the PCA.
\end{abstract}


\ifCLASSOPTIONonecolumn
  \pagenumbering{arabic}
\fi

\section{Introduction}\label{Sec:Intro}
The broadcast nature of wireless medium makes wireless communications especially vulnerable to various security threats such as eavesdropping, impersonating, and message modification. However, by establishing secret keys between legitimate terminals through a secret-key agreement (SKA) protocol, such threats can be efficiently nullified. To this end, information theoretic approaches \cite{Wyner75The-Wire-Tap:A, Ozarow84Wire-tap:A, Csiszar78Broadcast:A, Maurer_1993, Ahlswede_1993} for the SKA have been proposed and extensively studied. In \cite{Wyner75The-Wire-Tap:A}, Wyner considered a scenario called the wiretap channel in which an eavesdropper listens in on communications between legitimate terminals over a noisier channel than the one between legitimate terminals. In the seminar work, it was shown that a pair of legitimate terminals can share a secret key in total ignorance of the eavesdropper with group codes. Later, Csi\'{z}ar and K\"{o}rner \cite{Csiszar78Broadcast:A} generalized the Wyner's original work in \cite{Wyner75The-Wire-Tap:A}.  Motivated by the results, a great deal of research has been conducted on SKA over various types of wiretap channels \cite{Bloch2006LDPC-based:A, Andersson2012Secret-key:A, Koyluoglu10Polar:A, Wong11Secret-Sharing:A, Song12Perfect:A}. However, such SKA schemes seem impractical as most of studies on code design for wiretap channels are limited to the cases with asymptotically long block lengths \cite{Thangaraj07Applications:A, Mahdavifar11Achieving:A, Andersson10Nested:A} and the assumption of noisier eavesdropper's channel may not be guaranteed in many cases. In addition, to determine secrecy rates, eavesdropper's channel quality and/or statistics must be a priori known.

Meanwhile, it was shown in \cite{Maurer_1993} that a secret key can be shared with public discussion even if an eavesdropper has a better channel provided the legitimate terminals have knowledge of the eavesdropper's channel quality. In addition, a practical sequential SKA protocol was introduced by Maurer. The scheme is designed to sequentially perform the following three phases: the 
advantage distillation \cite{Maurer_1993}, information reconciliation \cite{Brassard94Secret-Key:A}, and privacy amplification \cite{Bennett95Generalized:A} phases.  The first phase, i.e. the advantage distillation, enables the legitimate terminals to share correlated random sequences with a higher correlation than the one an eavesdropper acquires even under the condition that neither of the legitimate terminals has an advantage compared to the eavesdropper. In the information reconciliation phase, the legitimate terminals make their correlated random sequences identical by exchanging information over public channel. Finally, each legitimate terminal independently performs the privacy amplification on the identical random sequence to generate a secret key  which the eavesdropper is completely ignorant of. The SKA scheme is adopted in the quantum key distribution (QKD) protocol \cite{Bennett84Quantum:A, Ekert91Quantum:A} in which the quantum entanglement is utilized to detect possible eavesdropping of the randomness sharing between two legitimate terminals. Due to the inherent advantage compared to the eavesdropper, the randomness sharing in the QKD protocol can be performed without the advantage distillation phase. After the randomness sharing, the protocol performs the information reconciliation and the privacy amplification. Meanwhile, in wireless communications, there have been a few notable efforts \cite{Liu12Exploiting:A, Ye10Information-Theoretically:A, Ren11Secret:A} 
to realize the randomness sharing by exploiting wireless channel reciprocity.


Recently, cellular systems with large antenna arrays
(LAAs\footnote{By the LAA, we mean 
a BS's antenna array of a number of 
antenna elements and this number is usually 
much larger than that of users in the cell under the service of 
the BS.}) have been extensively studied due to their attractive features 
\cite{Marzetta_1999, Marzetta_2010, Jose_2011, Khisti_2010, Xiong_2012, Im13Secret}. On one hand, a high spectral efficiency can be achievable
since small-scale fading and intra-cell interference can be efficiently mitigated by the LAA \cite{Marzetta_1999, Marzetta_2010, Jose_2011}. On the other hand, from a security point of view, the LAA is especially advantageous in the sense that a narrower beam formed by the LAA makes the reception at the passive eavesdropper significantly weakened. Thus, the secrecy rate of wiretap 
LAA channels grows with the number of transmit antennas \cite{Khisti_2010, Xiong_2012, Im13Secret, Geraci2012, Zhu2014}. 
Sum secrecy rates in multi-user MIMO system have been studied 
in \cite{Geraci2012}. Then, the research is further extended 
to a multi-cell setup \cite{Zhu2014}. These studies show that the LAA 
helps wireless systems to have an advantage over eavesdroppers, which is equivalent to performing the advantage distillation phase \cite{Bennett_1995} in the randomness sharing.

This work considers an SKA protocol for the system with LAA 
based on the sequential three-phase protocol. In particular, the randomness sharing is carried out in such a way that the 
base-station (BS)
first acquires a collection of channel state information (CSI) 
between the BS
and the multiple users from the
receptions of orthogonal training sequences simultaneously sent by the users during a fraction of coherence block. Then, the channel reciprocity \cite{Guillaud_2005} enables the BS to make a precoding vector for the subsequent downlink data transmission to each of users equipped with a single receive antenna. The BS transmits different random sequences weighted by precoding vectors to the legitimate users during the remaining fraction of the coherence time. The time-division duplex (TDD) mode
significantly reduces the channel estimation overhead \cite{Marzetta_1999}.

In this work, we assume that uncoded random sequences are transmitted 
and allows some errors to happen over the transmissions. After the downlink transmissions, the BS and users end up having correlated random sequences 
that are not necessarily identical due to possible errors in the received random sequences. However, the advantage associated with LAA enables the shared random sequences between the BS and legitimate users to have higher correlations than the ones at the passive eavesdroppers when the number of antennas at the 
BS is sufficiently large. Thus, by performing the subsequent information reconciliation and privacy amplification phases, the BS and 
users can have identical secret keys in the end. The key agreement protocol we consider in this work can be understood within the theoretical framework of \emph{sender excited} model \cite{Chou14, Wang14}.

However, recently, a serious security weakness of the SKA protocol was discussed by Zhou \emph{et al} in \cite{Zhou_2012} where it was pointed out that the precoding for legitimate users in the LAA-based TDD system is solely determined by CSI estimates based on the uplink training sequences which are exposed to active attacks. As a potential attack, the authors in  \cite{Zhou_2012} studied an active eavesdropping attack called \emph{pilot contamination attack} (PCA) in which an eavesdropper transmits the same pilot sequence as the one from a target user for the purpose of tilting the direction of beam towards 
the eavesdroppers. In particular, by contaminating the pilot sequence, the eavesdropper deceives the BS to make a precoding vector which steers the beam direction from the target user towards 
the eavesdropper, and the information sent by the BS leaks to the eavesdropper. Since the PCA was first introduced in \cite{Zhou_2012}, countermeasures to protect wireless communications from the PCA have not been well investigated, which motives our work. There have been efforts to detect the PCA in \cite{Im13Secret, Kapetanovic13Detection:A} among which the authors in \cite{Im13Secret} studied a PCA detector and an estimator of the eavesdropper's channel. While the detector utilizes statistics of the received signals at both BS and the target user, the estimator is based only on the ones at the BS side. Meanwhile, the work in \cite{Kapetanovic13Detection:A} proposed a PCA detection technique which employs random pilots from a set of phase-shift keying symbols. In \cite{Im13SKA}, 
an SKA protocol under potential PCA was proposed and evaluated. However, the SKA protocol in \cite{Im13SKA} is inefficient in the sense that the protocol requires multiple coherence blocks and simply discards suspicious packets.

In this paper, based on the three-phase sequential protocol we propose a modified SKA protocol tailored for nullifying the PCA with assumptions: 1) multiple non-colluding eavesdroppers equipped with a single antenna attempt the PCA to their own targets \cite{Xiong12AClosedForm, Romero-Zurita12Outage, Pinto12SecureII}, 2) the BS and legitimate users do not have any prior knowledge 
of the eavesdroppers such as the number of eavesdroppers in the network, their locations, and their transmit powers used in the PCA, and 3) other cells fully cooperate in a way not to use the training sequences of the users in the process of SKA. Thus, the pilot contaminations to the users in the SKA process come only from the active eavesdroppers employing the PCA if any. The cooperation with other cells can be justified since the SKA session needs to be rarely performed as compared to data transmission. Thus, the system throughput degradation due to the cooperation could be negligible. 

The sequential SKA protocol enables us to establish a secret key between legitimate parities even under the PCA if the SKA scheme has knowledge of how much information the active eavesdroppers have gained about the random sequences transmitted by the BS. However, unfortunately, such knowledge is not available.
To overcome the technical challenge, the standard sequential SKA protocol must be modified by 
introducing
a mechanism to estimate the amount of information leakage.

The main idea of the proposed SKA is inspired by the QKD protocol \cite{Bennett84Quantum, B92, E91}. The security of QKD is based on 
the principles of quantum mechanics, 
the no-cloning theorem \cite{Wootters82A-single:A} for example, implying that eavesdropper cannot overhear qubits transmitted from a transmitter to a receiver without introducing detectable anomalies. In the protocol BB84 \cite{Bennett84Quantum}, the legitimate terminals discuss a certain subset of their measurement results to detect the presence of eavesdropping. When eavesdropping is detected, the random sequence obtained from this 
session is discarded. We find that the wireless system with LAA has a similar property with which the presence of eavesdropper can be detected. That is, there is a complementary relation between the received signal strengths at the target user and eavesdropper. Once an eavesdropper attempts the 
PCA on a target user, the received signal strength at the target user becomes weaker than the one expected since the beam formed by the BS for the target user is partially steered towards the eavesdropper. Thus, the stronger the PCA, the wider gap between the signal strength measured at the target user and the one expected. In this paper, based on this relation, we 
derive an estimator to measure the CSI between the BS and 
eavesdroppers since the CSI is directly proportional to the amount of the information leakage. Then, the BS and the legitimate users can adjust the lengths of secret keys according to the estimated amount of information leakage contrary to the protocol BB84 \cite{Bennett84Quantum} where the generated secret keys are discarded when eavesdropping is suspicious. Performances of the proposed scheme are evaluated by conducting comparisons in numerical and analytic ways in various environments. The main contributions of this paper is summarized as follows: 
\begin{enumerate}
  \item We first introduce the proposed SKA protocol to defend wireless systems with LAA against the PCA. The impact of the PCA on the 
performance of the SKA scheme is analyzed in terms of average signal-to-interference-plus-noise power ratios (SINRs) at the target user and eavesdropper. The analysis results clearly show the complementary relation between the average SINRs at the target user and eavesdropper.
  
  \item Based on the complementary relation, 
we derive an estimator for the purpose of estimating the eavesdroppers' channels, i.e. the ones between the BS and eavesdroppers. It will be shown that the estimation results can be utilized for 
estimating the amount of information leakage during the randomness sharing.  
  
  \item We evaluate average secret key lengths when the 
estimate of information leakage is provided to the sequential SKA protocol which adaptively determines the length of resulting secret key.
Performance evaluations show that the secrecy outage probability, i.e. the probability not to achieve perfect secrecy, decreases exponentially fast with 
the number of antennas. In addition, it will be shown that 
a stronger PCA ironically results in a better system performance, i.e. 
a lower secrecy outage probability due to the complementary relation.
  \item Comprehensive performance evaluations are carried out to see trade-offs between the outage probability and average secret key length with different combinations of system parameters, such as the number of users, the number of antennas at the BS, and the lengths of random sequences. The results of performance evaluations enable the system designer to choose appropriate parameters to meet various system requirements.
\end{enumerate}

The rest of the paper is organized as follows. In Section \ref{Sec:SKA}, we describe the scenario under investigation which includes the proposed SKA protocol, an adversary model, and channel models. In Section \ref{Sec:EstInfLeak}, the complementary relation of the received signal strengths is analyzed by investigating average SINRs at the legitimate user and eavesdropper. Based on the relation, an estimation scheme for the amount of information leakage to eavesdropper is proposed and analyzed. In Section \ref{Sec:Results}, performance evaluations for the proposed estimation scheme are carried out in terms of a normalized mean-square error. In addition, the secrecy outage probability and the average length of secret key are extensively evaluated with various combinations of design parameters. Finally, we make conclusions in Section \ref{Sec:Con}.

\emph{Notation:} 
Bold face upper and lower case letters are used to denote matrices and vectors, respectively. Transpose and Hermitian are denoted by $(\cdot)^T$ and $(\cdot)^\dag$, respectively. We use $[x]^+$ for $\max\{0,x\}$. $\mathbf{0}_M$  and $\mathbf{I}_M$ denote the $M \times 1$ all zero vector and $M \times M$ identity matrix, respectively. We use $\norm \cdot$ for the Euclidean vector norm. A probability density function (pdf) and a conditional pdf are denoted by $f(\cdot)$ and $f(\cdot|\cdot)$, respectively. When the pdf and the conditional pdf are parameterized by an unknown parameter $\theta$, they are denoted by $f(\cdot ; \theta)$ and $f( \cdot|\cdot ; \theta)$, respectively.  
%
%
\section{Secret Key Agreement Scheme} \label{Sec:SKA}
\subsection{System model}
\begin{figure*}[t]
	\centering
		\subfigure[Pilot contamination attack for $K=2$ and $K_e=1$]
		{
			\includegraphics[width=0.45\textwidth]{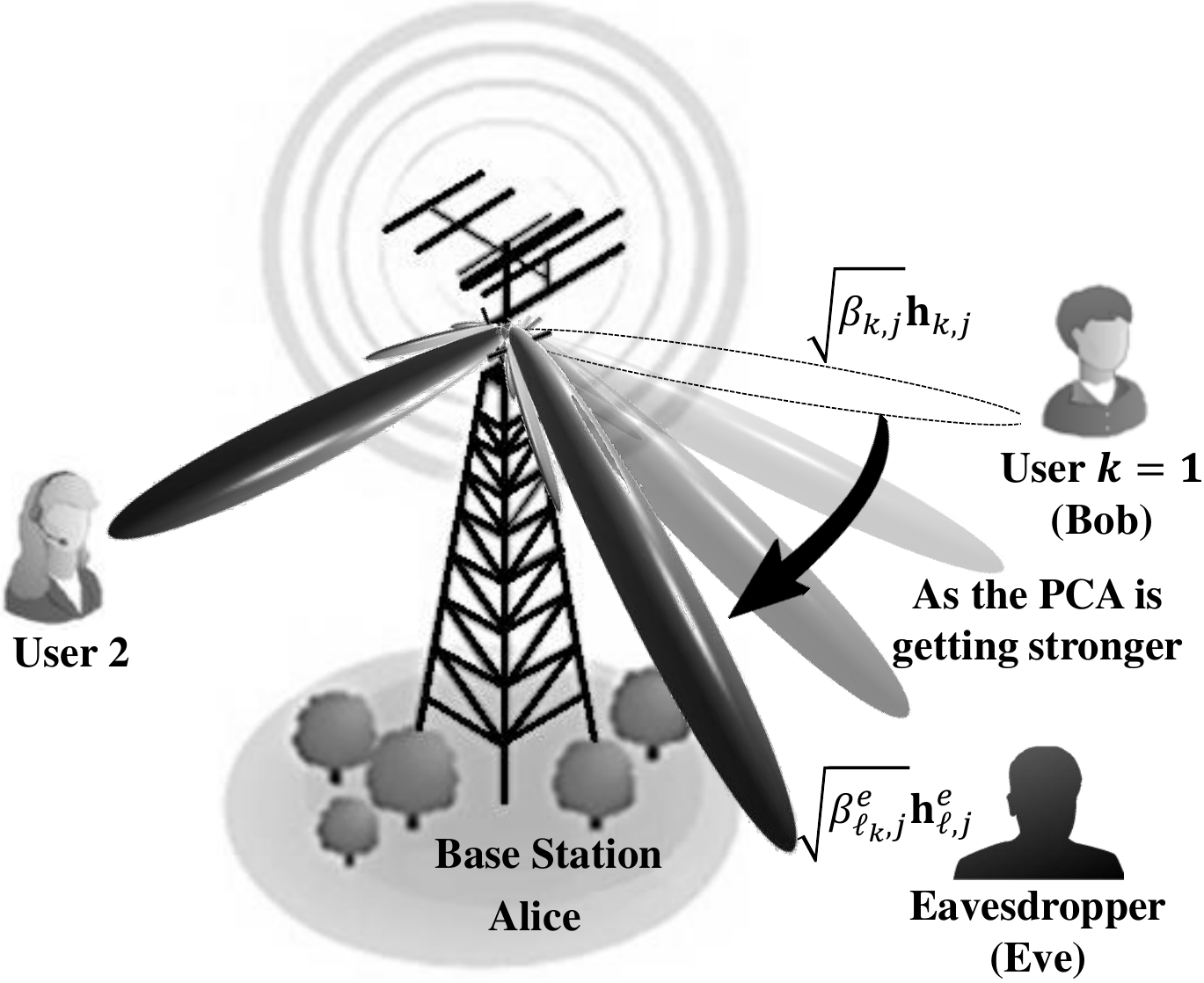}
			\label{Fg:Cell}
		}
		\subfigure[Secret key generation protocol]
		{
			\includegraphics[width=0.5\textwidth]{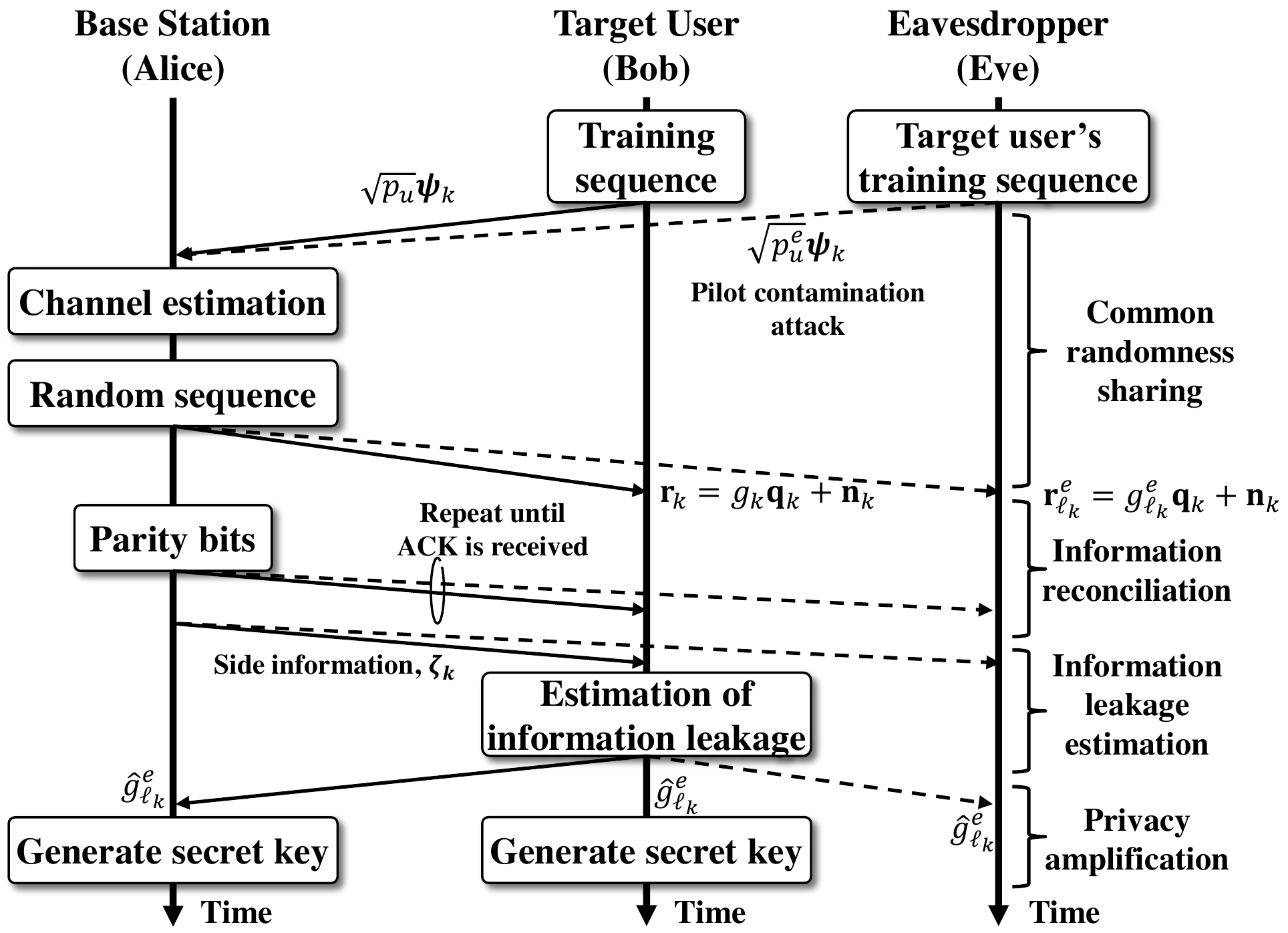}
			\label{Fg:SKA}
		}		
	\caption{System model}
	\label{Fg:Fig1}
\end{figure*}

We consider a TDD-based cellular system where a base station in each cell, called \emph{Alice}, aims at establishing different secret keys with $K$ legitimate users in the presence of $K_e$ active eavesdroppers. 
Alice has an array of a large number of antennas, say $M$ antennas, while each of $K$ users has a single antenna. Fig. \ref{Fg:Cell} illustrates an example where the $k$-th legitimate user is in the SKA session over a wireless channel, and the $\ell_k$-th eavesdropper attempts the PCA to the $k$-th legitimate user over a different wireless channel. 

The wireless channels in this work experience
both small-scale and large-scale fading. In Fig. \ref{Fg:Cell}, the channel realization at the $j$-th coherence block between Alice and the $k$-th user is given by $\sqrt{\beta_{k, j}}\mathbf{h}_{k, j}$.  
Here, $\sqrt{\beta_{k, j}}$ accounts for the nonnegative large-scale fading factor determined by path-loss and shadowing, which are slowly varying over time, while $\mathbf{h}_{k, j}$ is an $M \times 1$ vector representing the small-scale fading and varying faster than the large-scale fading factors. We assume that the large-scale fading factors, $\sqrt{\beta_{k, j}}$'s for $1 \le k \le K$ are public information a priori known to everyone including eavesdroppers. Meanwhile, it is assumed that the small-scale fading factor, $\mathbf{h}_{k, j}$'s follow $\mathcal{CN}(\mathbf{0}_M,\mathbf{I}_M)$ and are statistically independent and identically distributed (i.i.d.). In addition, we assume block fading channels, i.e.  $\mathbf{h}_{k, j}$'s are static over a coherence block and i.i.d. across coherence blocks. The proposed SKA protocol is performed within one coherence block, and thus we will omit the coherence block index $j$ for simplicity throughout this paper. 
Contrary to the large scaling fading factors, 
only the statistical properties of the small scale fading factors
are known to everyone. Thus, each realization of $\mathbf{h}_k$ 
is unknown and must be estimated if needed. Similarly, 
the wireless channels between Alice and the eavesdroppers are 
modeled as $\sqrt{\beta^e_{\ell}}\mathbf{h}^e_{\ell}$, $\ell \in 
\{1,\ldots, K_e\}$, where $\sqrt{\beta^e_{\ell}}$ and $\mathbf{h}^e_{\ell}$ are the large- and small-scale fading factors, respectively. Note that the 
coherence block index $j$ is omitted as aforementioned. Contrary to the legitimate users' channels, it is assumed that the large scaling fading factors $\beta^e_\ell$ is known only to the eavesdroppers, i.e. not available to the legitimate users and Alice. 

For the cellular system, we consider the proposed SKA protocol summarized
in Fig. \ref{Fg:SKA}, which consists of 
common randomness sharing, information reconciliation, 
information leakage estimation and privacy amplifications. In this section, the  building blocks of the proposed scheme in Fig. \ref{Fg:SKA} are introduced in detail except for the information leakage estimation which will be discussed in Section \ref{Sec:EstInfLeak}.

\subsection{Common Randomness Sharing (CRS)}
\subsubsection{Uplink Training}
The CRS is initiated by users who want to establish secret keys. As the first step of the CRS, $K$ users simultaneously transmit orthonormal training sequences at the beginning of a coherence block  so that Alice can estimate CSI of each user, i.e. $\mathbf{h}_k$. In particular, the legitimate users transmit $\sqrt{p_u N_u} \boldsymbol{\psi}_k$ to Alice,
where $p_u$ is the uplink training power, $\boldsymbol{\psi}_k$ is a $1\times N_u$ binary orthonormal training sequence, i.e. $\boldsymbol{\psi}_k\boldsymbol{\psi}^{\dagger}_k=1$, $\boldsymbol{\psi}_k\boldsymbol{\psi}^{\dagger}_{\ell}=0$ for $k \neq \ell$, and $k$ and $\ell$ are in  $\mathcal{K} = \{1, 2, \ldots, K\}$. 
Here, $N_u$ (i.e., the length of orthonormal training sequences)
is usually larger than or equal to $K$.

Meanwhile, for the PCA, the eavesdroppers inject their target users' training sequences perfectly synchronized with the uplink training sequences originated from the legitimate users.  Then, the received signal at Alice becomes
\ifCLASSOPTIONonecolumn 
\begin{equation}\label{eq:Y}
	\mathbf{Y} = 
	\sum_{k\in\mathcal{K}} \sqrt{p_u \beta_k N_u}  \mathbf{h}_{k}  \boldsymbol{\psi}_k + \sum_{\ell \in \mathcal{E}} \sqrt{p^e_\ell \beta^e_\ell N_u} \mathbf{h}^e_\ell \boldsymbol{\psi}_\ell +  \mathbf{U}, 
\end{equation}	
\fi
\ifCLASSOPTIONtwocolumn 
\begin{align}\label{eq:Y}
	\mathbf{Y} = 
	\sum_{k\in\mathcal{K}} \sqrt{p_u \beta_k N_u}  \mathbf{h}_{k}  \boldsymbol{\psi}_k + \sum_{\ell \in \mathcal{E}} \sqrt{p^e_\ell \beta^e_\ell N_u} \mathbf{h}^e_\ell \boldsymbol{\psi}_\ell +  \mathbf{U}, 
\end{align}	\fi
where $\mathcal{E} = \{1, \ldots, K_e\}$ is 
the index set of eavesdroppers, $p^e_\ell$ is the PCA power of the $\ell$-th eavesdropper, $\mathbf{U}$ is an $M \times N_u$ noise matrix in which each entry is independent zero-mean circularly-symmetric complex Gaussian (CSCG) with unit variance.

Since we assume multiple non-colluding eavesdroppers in a cell, it is conceivable that multiple eavesdroppers may perform the PCA to a target user. This is however of no benefit to the non-colluding eavesdroppers since they are in a competition to pull the user beam towards them, and thereby the amount of information leaked to each eavesdropper decreases. Furthermore, if an eavesdropper performs the PCA on multiple users simultaneously, it will receive a superposition of multiple
signals and its eavesdropping performance becomes interference-limited as when a signal to a user is to be detected or decoded, the other signals (to the other users) become interfering signals.
 
Thus, throughout this paper, we consider the best case scenario for eavesdroppers as follows:
\begin{itemize}
  \item An eavesdropper does not attempt the PCA targeting more than one user 
at a time, and 
  \item A user is not targeted by more than one eavesdropper at a time.
\end{itemize}
The assumptions guarantee that there are at most $K$ eavesdroppers in a cell, i.e. $K_e \le K$, and each of them has its unique target. In this work, without loss of generality, we assume that $\mathcal{K} = \mathcal{E}$, i.e. $K_e = K$ and the $k$-th legitimate user is attacked by the $k$-th eavesdropper. Hereafter, the $k$-th user and the $k$-th eavesdropper are called Bob and Eve for short when there is no risk of confusion.

Since Alice does not know the CSI for the legitimate users, 
she has to estimate the CSI based on the received signal, $\mathbf{Y}$. Due to the orthonormality, $\mathbf{y}_k = \mathbf{Y} \boldsymbol{\psi}^\dag_k$ is a sufficient statistic for estimating the CSI for Bob, i.e. $\mathbf{h}_k$, 
which is expressed as
\begin{equation} \label{eq:y_k}
  \mathbf{y}_k = \sqrt{c_k} (\mathbf{h}_k + w_k \mathbf{h}^e_k) + \mathbf{u}_k, \quad \text{for } k \in \mathcal{K},
\end{equation} 
where $c_k = p_u \beta_k N_u$, $\mathbf{u}_k = \mathbf{U} \boldsymbol{\psi}_k^\dag$, and
\begin{equation} \label{eq:w_k}
  w_k =  \sqrt{\frac{p^e_k \beta^e_k}{p_u \beta_k}}, \text{ for } k \in \mathcal{K}.
\end{equation}
In \eqref{eq:w_k}, $ w_k \in [0, \infty)$\footnote{We exclude $w_k > 1$ since an eavesdropper considered in this paper aims to eavesdrop a secret key between Alice and a target without revealing its presence. Nevertheless, the eavesdropper may increase its uplink training power by $w_k > 1$. In this case, the attack can be detected by the target user immediately \cite{Im13Secret}, and Alice and the target user can avoid such an attack by establishing a new wireless channel to generate a secret key.} represents  \emph{the effective strength of the PCA} to Bob, and the case of $w_k = 0$ implies passive eavesdropping.
We employ a minimum mean-square-error (MSE), or MMSE estimator to estimate $\mathbf{h}_k$ \cite{Poor_1994, ChoiBook} which is 
\begin{equation}
  \hat{\mathbf{h}}_k = \frac{\sqrt{c_k}}{ 1+ \left(1+w_k^2 \right) c_k } \mathbf{y}_k.
  \label{eq:estimate_h_k}
\end{equation}
Note that the estimator in \eqref{eq:estimate_h_k} requires the knowledge of $w_k$ that is, however, not available to Alice. 
Thus, when Alice is not aware of the PCA, she has the estimate, $\hat{\mathbf{h}}_k$ in \eqref{eq:estimate_h_k} with $w_k=0$, which becomes
\[
  \hat{\mathbf{h}}_k = \mathbf{y}_k \sqrt{c_k}/(1 +  c_k).
\]

\subsubsection{Downlink Transmission}
In the downlink transmission, Alice generates $K$ binary random sequences
of length $N_b$, denoted by
$\mathbf{b}_k = [b_{k, 1}, \ldots, b_{k, N_b}]^T$ for $k \in \mathcal{K}$,
which are then mapped into a modulated sequence of length $N_d$,
denoted by
$\mathbf{q}_k = [q_{k, 1}, \ldots, q_{k, N_d}]^T$, 
where $q_{k, j} \in \mathbb{C}$, $1 \le j \le N_d$. 
Alice simultaneously sends 
$\mathbf{q}_k$ weighted by precoding vectors in the form of $\sqrt{p_d} \sum_{k=1}^K \mathbf{a}_k \mathbf{q}_k^T$ for all $k \in \mathcal{K}$, where $p_d$ is the downlink transmission power, 
and $\mathbf{a}_k $ is an 
$M \times 1$ precoding vector for Bob. The average power of 
$\mathbf{q}_k$ is normalized to be
$\frac{1}{N_d} \mathbb{E} [||\mathbf{q}_{k}||^2]=1$.
The precoding vector $\mathbf{a}_k$ is determined as a function of $\hat{\mathbf{h}}_k$, i.e $\mathbf{a}_k= \varphi ( \hat{\mathbf{h}}_k)$ where $\varphi (\cdot)$ is a precoding vector generating function that could be chosen in different ways. 
In this paper, we consider the matched filter (MF) precoding to gain more insights into the proposed system\footnote{In general, two linear precoding schemes, MF and MMSE precoding methods, are of practical interest \cite{Poor_1994}. Our main results are also valid with the MMSE precoding, but its complicated expression may make it harder to understand essentials.}:
\begin{align}\label{eq:MF_precoder}
	\mathbf{a}_k &=\frac{\hat{\mathbf{h}}_k}{ \|{\hat{\mathbf{h}}_k} \| },  
\quad \text{for } k \in \mathcal{K}.
\end{align}

Meanwhile, at the receiver side, without loss of generality, we assume that the received signal of each user is normalized by $\sqrt{p_d \beta_k M}$ for $k \in \mathcal{K}$. 
Then, the normalized received signal vector for Bob is given by 
\begin{align}
	{\mathbf{r}}_k &= \left( \frac{{ \mathbf{h}_k^T \mathbf{a}_k }}{\sqrt{M}}  \right) \mathbf{q}_k+\sum_{\ell \neq k} \left( \frac{{ \mathbf{h}_k^ T \mathbf{a}_\ell  }  }{\sqrt{M}} \right) \mathbf{q}_\ell+ \mathbf{z}_k  \,\, \text{for } k,\ell \in \mathcal{K}, \label{eq:r_Bob}
\end{align}
where $\mathbf{z}_k$ is an $N_d \times 1$ zero-mean CSCG noise vector 
with covariance matrix $(p_d \beta_k M)^{-1} \mathbf{I}_N$. 
We define an \emph{effective downlink channel gain} (EDCG) from Alice to Bob as $g_k =\frac{1}{\sqrt{M}}{\mathbf{h}_k^T \mathbf{a}_k}$ which leads to the following simplified expression of $\mathbf{r}_k$:
\begin{align}
	{\mathbf{r}}_k=  g_k \mathbf{q}_k+ \mathbf{n}_k \quad \text{for } k \in \mathcal{K}, \label{eq:r_Bob2}
\end{align}
where $\mathbf{n}_k = \sum_{\ell \neq k} (\frac{{\mathbf{h}_k^ T \mathbf{a}_\ell} }{\sqrt{M}}) \mathbf{q}_k+ \mathbf{z}_k$.
In Appendix \ref{App:A}, we show that $\mathbf{h}_k^ T \mathbf{a}_\ell$ for $k \neq \ell$ follows $\mathcal{CN}(0,1)$. 
Thus, the last term, $\mathbf{n}_k$ follows $\mathcal{CN} (\mathbf{0}_{N_d}, \sigma_{n_k}^2 \mathbf{I}_{N_d})$ where 
\begin{equation}
\sigma_{n_k}^2 = \frac{1}{M}\left(K - 1 + \frac{1}{p_d\beta_k}\right).
	\label{EQ:sn}
\end{equation}

Following a similar approach, we have the normalized received signal vector for Eve as 
\begin{align}\label{eq:r_Eve}
	\mathbf{r}^e_k &=  \left( \frac{{{\mathbf{h}^e_k}^T \mathbf{a}_k}}{\sqrt{M}} \right) \mathbf{q}_k + \sum_{\ell \neq k} \left( \frac{{{\mathbf{h}^e_k}^ T \mathbf{a}_\ell}}{\sqrt{M}} \right) \mathbf{q}_\ell \quad \text{for } k, \ell \in \mathcal{K}.
\end{align}
Note that no noise is assumed in the normalized received signal vector for Eve, which leads to an unfavorite 
scenario to Bob and Alice. 
Thus, the sum of interference due to the 
downlink signals to the other users is only the factor 
to impair recovering $\mathbf{q}_k$ from $\mathbf{r}^e_k$. In addition, the assumption allows us not to care for the  the locations of the eavesdroppers. The received signal at Eve, $\mathbf{r}^e_k$ can also be simplified to 
\begin{align}
	\mathbf{r}^e_k & = g^e_k \mathbf{q}_k+ \mathbf{n}^e_k  \quad \text{for } k \in \mathcal{E},\label{eq:r_Eve2}
\end{align}
where $g^e_k =\frac{1}{\sqrt{M}}{{\mathbf{h}^e_k}^T \mathbf{a}_k}$ is the EDCG from Alice to Eve, and $\mathbf{n}^e_k = \sum_{k \neq \ell} ( \frac{{{\mathbf{h}^e_k}^ T \mathbf{a}_\ell} }{\sqrt{M}}  ) \mathbf{q}_\ell$ whose distribution is given by $\mathcal{CN}(\mathbf{0}_{N_d}, \sigma_{n_{e,k}}^2 \mathbf{I}_{N_d})$ with $\sigma_{n_{e,k}}^2 = (K - 1)/M$. In Fig. \ref{Fg:SKA}, the received signals $\mathbf{r}_k$ and $\mathbf{r}^e_k$ at Bob and Eve, respectively, are depicted as the results of the CRS.


\subsection{Information Reconciliation and Privacy Amplification}

After the CRS, secret keys are extracted from the shared randomness 
through the information reconciliation and 
privacy amplification phases, which requires knowledge of information leakage 
to eavesdroppers by the PCA. However,  we  assume that Alice and the $K$ users have no information about eavesdroppers. In this section, we first review the information reconciliation and privacy amplification phases shown in Fig. \ref{Fg:SKA} and then discuss how such missing information hinders the phases from generating secret keys.

In the information reconciliation phase, 
Alice sends parity bits  over an \emph{authenticated public channel} to the users who need to correct errors occurred in the CRS. 
In the CRS, Alice transmits to Bob a sequence $\mathbf{q}_k$ 
that arrives at Bob and Eve who have received sequences  $\mathbf{r}_k = (r_{k,1} , \ldots , r_{k,N_d})$ and $\mathbf{r}^e_k = (r^e_{k,1}, \ldots, r^e_{k,N_d})$, respectively. After the randomness sharing, Alice and Bob can have shared information which amounts to $I(\mathbf{Q}_k; \mathbf{R}_k)$ where $\mathbf{Q}_k$ and $\mathbf{R}_k$ are random vectors corresponding to the realizations $\mathbf{q}_k$ and $\mathbf{r}_k$, respectively. The uncertainty between $\mathbf{Q}_k$ and $\mathbf{R}_k$ must be resolved to make them identical in the information reconciliation phase. According to the Slepian-Wolf theorem \cite{Slepian73Noiseless}, the information reconciliation requires at least $\nu_k = H(\mathbf{Q}_k | \mathbf{R}_k) = H (\mathbf{Q}_k) - I (\mathbf{Q}_k ; \mathbf{R}_k)$ bit exchanges over 
the public channel. 
While either Alice or Bob can generate and transmit the parity bits, this work assumes that Alice sends the parity bits since Alice already has the estimated CSI. Then, Alice and Bob can have an identical sequence, $\mathbf{q}_k$ which turns into the binary sequence, $\mathbf{b}_k$ of length $N_b = H (\mathbf{Q}_k)$. However, the sequence is not secure from eavesdropping due to the information leakage, denoted by $\mathbf{E}_k$, during the first two phases, i.e. the CRS and information reconciliation phases. Alice and Bob extract a secret key from $\mathbf{Q}_k$ by performing the privacy amplification process which eliminates the amount of eavesdropped information, $\mathbf{E}_k$, from $H(\mathbf{Q}_k)$. This can be done by using a hash function, $G_k \in\mathcal{G}: \{0, 1 \}^{N_b} \to \{0,1\}^{s_k}$, randomly chosen from a family of universal hash functions, $\mathcal{G}$ \cite{Carter79Universal} where $s_k$ is the length of secret key, $\mathbf{S}_k$. If we determine $s_k$ such that 
\begin{equation}
	s_k = \left[ N_d \{  I(Q_k; R_k) - I(Q_k; R^e_k) \} - 2a_k - 2 - b_k \right]^+, \label{eq:Key1}
\end{equation} 
where $Q_k$ and $R^e_k$ are the i.i.d. 
random variables representing the components in the random vectors, $\mathbf{Q}_k$ and $\mathbf{R}^e_k$, respectively, it is guaranteed that, for a sufficiently large $N_d$, the eavesdropper's uncertainty about the secret key,
denoted by $H(\mathbf{S}_k | G_k, \mathbf{E}_k)$, is bounded by
\begin{align}
	H (\mathbf{S}_k | G_k, \mathbf{E}_k) \ge s_k - \frac{2^{-b_k}}{\ln 2} \text{ with probability } 1-2^{-a_k}. \label{eq:Key2} 
\end{align}
By appropriately choosing $a_k$ and $b_k$, we can obtain a sufficiently long secret key with the eavesdropper's uncertainty 
that can be arbitrarily close to $s_k$, which implies perfect secrecy. 

However, the standard sequential SKA protocol cannot be directly applied to our scenario due to the assumption, i.e no prior knowledge 
about the eavesdroppers. 
In particular, the problems are as follows: 
\begin{enumerate}
	\item Unknown $\nu_k$: The required $\nu_k$ for the information reconciliation is derived from $I(Q_k ; R_{k} ) $ which is a function of CSI between Alice and Bob. However, once the training sequence sent by Bob is contaminated by the PCA, Alice would have a poor
estimate of the CSI, and its corresponding $\nu_k$ may not be enough for Bob to decode $\mathbf{q}_k$.
	\item Unknown $I(Q_k; R^e_k)$: Due to unknown CSI between Alice and Eve, Alice and Bob cannot measure $I(Q_k ; R^e_k)$. Thus, the standard sequential SKA protocol cannot determine the length of secret key in \eqref{eq:Key1} to achieve perfect secrecy.
\end{enumerate}

The first problem of unknown $\nu_k$ can be resolved by exchanging additional messages between Alice and Bob. For example, if a rateless Slepian-Wolf code is employed, Alice can repeatedly send 
additional parity bits to Bob until Bob has a 
sufficiently large number of parity bits to make $\mathbf{r}_k$ identical to $\mathbf{q}_k$ and sends an acknowledge message to Alice. In this case, although a certain amount of overhead is inevitable, the rateless Slepian-Wolf codes can be a practical solution  since they do not require  instantaneous CSI between a transmitter and a receiver. 
Thus, throughout this paper, we assume that Bob can perfectly recover $\mathbf{q}_k$ using a rateless Slepian-Wolf code. 
The practical design and optimization of rateless Slepian-Wolf codes are presented in \cite{Eckford05Rateless, 4557405}.

However, unknown $I(Q_k; R^e_k)$ is still problematic. Since the eavesdroppers will not reveal their presence and information about the PCA, a mechanism 
must be introduced into the proposed scheme to 
estimate $I(Q_k; R^e_k)$, 
which is realized in this work by taking advantage of a trace left by the eavesdropper during the PCA. To be precisely, Bob is suspicious of the PCA if its received signal strength unexpectedly drops during the CRS. 
Then, Bob can estimate $I(Q_k; R^e_k)$ by comparing the received signal strength with its expectation. As depicted in Fig. \ref{Fg:SKA}, the estimation 
results, denoted by $\hat{g}^e_k$, are provided to the privacy amplification, 
which in the end generates secret keys as the final results of the proposed SKA protocol. The details of the proposed method  will be introduced and analyzed in the following sections. 

Before finishing this section, a few practical issues should be discussed. In fast fading environments, the proposed SKA protocol may result in a considerable overhead to establish a secret key at a sufficiently long length. 
In such a case, there might be a trade-off 
between the data throughput and key renewal rate for a given
secret key length. 
In addition, an unexpected drop of 
the received signal strength can happen due to user mobility or sudden environment changes, which may introduce a false alarm and thus degrade 
the performance of the proposed SKA protocol, i.e. the average secret key length.

\section{Estimation of Information Leakage} \label{Sec:EstInfLeak}

In this section, we first investigate an inevitable
complementary relation between SINR's at Bob and Eve, i.e. 
the increase of SINR at one party 
must result in the decrease of SINR at the other. 
Based on this complementary relation, we propose an estimator 
of the EDCG from Alice to Eve (i.e. $g^e_k$), which allows us
to estimate the information leakage to Eve, $I(Q_k; R^e_k)$. 
The crucial complementary relation also promises 
a better estimation result for a stronger PCA.

\subsection{The Impact of PCA on Average SINRs}

The average SINRs at Bob and Eve,
denoted by $\text{SINR}_k $ and $\text{SINR}^e_k$, are defined as 
\ifCLASSOPTIONonecolumn 
\begin{align*}
	\text{SINR}_k = \frac{ \mathbb{E} \left[ (g_k \mathbf{q}_k)^\dag (g_k \mathbf{q}_k )\right] }{\mathbb{E} \left[ {\mathbf{n}_k}^\dag \mathbf{n}_k  \right] } 
	\text{ and } 
	\text{SINR}^e_k =  \frac{ \mathbb{E} \left[ (g^e_k \mathbf{q}_k )^\dag (g^e_k \mathbf{q}_k )\right] }{\mathbb{E} \left[ {\mathbf{n}^e_k}^\dag \mathbf{n}^e_k  \right] }.
\end{align*}
\else
\begin{align*}
	\text{SINR}_k & = \mathbb{E}[ (g_k \mathbf{q}_k )^\dag (g_k \mathbf{q}_k )]/ (N_d \mathbb{E}[\mathbf{n}_k^\dag \mathbf{n}_k]), \text{ and } \\
	\text{SINR}^e_k & =  \mathbb{E} [(g^e_k \mathbf{q}_k )^\dag (g^e_k \mathbf{q}_k )]/(N_d \mathbb{E} [{\mathbf{n}^e_k}^\dag \mathbf{n}^e_k]).
\end{align*}
\fi
Then,  the average SINRs are derived in the following result. 

%

\begin{thm}\label{thm:SINR} When Bob and Eve receive the signals in  \eqref{eq:r_Bob2} and \eqref{eq:r_Eve2}, respectively, their SINR's are given by 
\ifCLASSOPTIONonecolumn 
\begin{align}\label{eq:SINR}
	\text{SINR}_k   = \frac{M c_k + w_k^2 c_k +1 }{(1 + (1+w_k^2) c_k)( K-1 + \frac{1}{p_d \beta_k})} \text{ and }
	\text{SINR}^e_k = \frac{M w_k^2 c_k + c_k  + 1 }{(1 + (1+w_k^2) c_k) (K - 1)}.
\end{align} 
\else
\begin{align}\label{eq:SINR}
	\text{SINR}_k   & = \frac{M c_k + w_k^2 c_k + 1}{(1 + (1 + w_k^2) c_k)( K - 1 + \frac{1}{p_d \beta_k})} \text { and} \\
	\text{SINR}^e_k & = \frac{M w_k^2 c_k + c_k + 1}{(1 + (1 + w_k^2) c_k) (K - 1)}.
\end{align} 
\fi
\end{thm}
\begin{proof}
See Appendix \ref{App:SINR}.  
\end{proof}

\begin{rmk}
By letting $w_k \to 0$, the results in Theorem \ref{thm:SINR} turn into the ones for passive eavesdropping. Then, we have $\text{SINR}_{k}= \frac{M c_k +1 }{( 1 + c_k ) ( K-1 + \frac{1}{p_d \beta_k})}$ and  $\text{SINR}^e_k = (K - 1)^{-1}$, which shows that $\text{SINR}_{k}$ grows with the 
number of antennas, $M$. On the other hand, $\text{SINR}^e_k$ does not depend 
on $M$. Thus, the length of secret key can be increased by employing more antennas at Alice under the passive eavesdropping. While similar results can be found in \cite{Khisti_2010, Xiong_2012, Im13Secret}, the results in the work are different from the previous ones in that the signals for the other $K - 1$  users in our model act as interference and preclude the eavesdropper from taking information. 
\end{rmk}

The ratio of SINR's at Bob and Eve is readily approximated as $\gamma_k = \text{SINR}_k/\text{SINR}^e_k \approx 1/w^2_k$ 
when $K \gg 1$ and $M \gg 1$. The ratio clearly shows that the SINR at Bob is inversely proportional to the effective strength of the PCA, $w_k$, which leads to a better SINR at Eve. While the results in Theorem \ref{thm:SINR} describe the average behavior of SINR, the instantaneous amount of information leakage in each SKA protocol is not provided. In the next subsection, bearing the complementary relation in mind, we will derive
an estimator of the EDCG from Alice to Eve, which in turn gives us 
an instantaneous estimate of the information leakage to eavesdroppers.

\subsection{Estimation of Eavesdropper's Channel}
According to \eqref{eq:r_Eve}, the amount of information leaked to Eve, 
$I(Q_k ; R^e_k)$, is readily found by Bob when the EDCG from Alice to Eve, 
$g^e_k$ is known to Bob. To this end, in this subsection, 
we derive an estimator for $g^e_k$. 
The estimation of $g^e_k$ is carried out by fulfilling a series of estimations: 1) maximum-likelihood estimation (MLE) for $w_k$ 2) MMSE estimation for $g_k$,  and 3) MMSE estimation for $g^e_k$, 
where their estimates are denoted by $\hat{w}_k$, $\hat{g}_k$, and $\hat{g}^e_k$, respectively. The estimation results from the first two steps are used as unknown parameters in the estimations of $g^e_k$.

The main idea behind the proposed estimator is to exploit the complementary relation between the received signal strengths at Bob and Eve, which results in an unexpected drop of signal strength at the target user when the PCA is attempted to Bob. The difference between the received signal strengths and his expectation will be used to estimate $g^e_k$. However, the CSI of the channel from Alice to Bob is unknown to Bob and thus he does not know his expected signal strength. To resolve this issue, as shown in Fig. \ref{Fg:SKA}, Alice sends side information $\zeta_k = \norm {\mathbf{y}_k }$ for $ k \in \mathcal{K}$, i.e. the strengths of her received signals right after the information reconciliation phase. 


\subsubsection{Estimation of $w_k$}
With the side information $\zeta_k$ and the normalized received signal vector at Bob in \eqref{eq:r_Bob}, we can derive 
the MLE, $\hat{w}_k$, as
\begin{align}\label{eq:mle_w}
 	\hat{w}_k = \arg \max_{w_k} f(\mathbf{r}_k | \mathbf{q}_k, \zeta_k ; w_k),
\end{align}
where $\mathbf{q}_k$ is given by the information reconciliation 
using a rateless Slepian-Wolf code. The pdf $f(\mathbf{r}_k | \mathbf{q}_k, \zeta_k; w_k)$ in \eqref{eq:mle_w} can be factorized as  
\ifCLASSOPTIONonecolumn 
\begin{align}\label{eq:pdf_r}
f(\mathbf{r}_k | \mathbf{q}_k, \zeta_k; w_k)
	& =  \int f({\mathbf{r}}_k | \mathbf{q}_k, \zeta_k, g_k; w_k) f({g}_k | \mathbf{q}_k, \zeta_k; w_k) d{g}_k\nonumber\\
	& \mathop = \limits^{(a)} \int f({\mathbf{r}}_k | \mathbf{q}_k, g_k ;w_k) f(g_k| \zeta_k; w_k) dg_k,
\end{align}
\else 
\begin{align}\label{eq:pdf_r}
&f (\mathbf{r}_k | \mathbf{q}_k, \zeta_k ; w_k) \nonumber\\
	& = \int{f({\mathbf{r}}_k | \mathbf{q}_k, \zeta_k, g_k ;w_k) f({g}_k | \mathbf{q}_k , \zeta_k; w_k) d{g}_k } \nonumber\\
	& \mathop = \limits^{(a)} \int{f({\mathbf{r}}_k | \mathbf{q}_k, g_k ;w_k) f(g_k | \zeta_k; w_k) d{g}_k },
\end{align}
\fi
where $(a)$ is due to the facts that $\mathbf{r}_k$ is independent of $w_k$ and $\zeta_k$ for a  given $g_k$, and $g_k$ is independent of $\mathbf{q}_k$. 

Note that finding $\hat{w}_k$ by substituting \eqref{eq:pdf_r} into \eqref{eq:mle_w} requires an exhaustive search as there is no closed-form 
solution. It may be impractical to find 
the estimate by performing numerical integrations, especially when the users suffer from limited computing power and/or power resources. Thus, we 
approximate the MLE of $w_k$ by taking one of the key features of systems with LAA, i.e. the randomness caused by fading and noise vanishes as the number of antennas increases \cite{Marzetta_2010, Jose_2011}.  The following theorem describes the asymptotic behavior of $g_k$ as $M \to \infty$.
\begin{thm}\label{thm:convergence}
Under the PCA with $w_k$, as $M$ increases, $g_k$ converges to $\mu_{g_k} =\frac{\sqrt{c_k}}{1 + (1 + w_k^2) c_k} \frac{\zeta_k}{\sqrt{M}}$ in probability for a given $\zeta_k$.
\end{thm}
\begin{proof}
See Appendix \ref{App:Limit_EDCG}.
\end{proof}

Thus, by applying Theorem \ref{thm:convergence} to \eqref{eq:pdf_r}, we have 
$f(\mathbf{r}_k | \mathbf{q}_k, g_k ; w_k) \to f(\mathbf{r}_k | \mathbf{q}_k, \mu_{g_k}; w_k) $ in probability for a large $M$.
This result allows us to have a simpler MLE of $w_k$. 
That is, the value of $w_k$ for which the derivative of $f(\mathbf{r}_k | \mathbf{q}_k, \mu_{g_k}; w_k) $ with respect to $w_k$ is equal to zero corresponds to $\hat{w}_k$. After some manipulations, we obtain a closed form expression 
for $\hat{w}_k$ as 
\begin{align}\label{eq:MLE_wk}
	\hat{w}_k = \sqrt{ \left[ \frac{\zeta_k}{\frac{\mathbf{r}_k^\dag \mathbf{q}_k}{\mathbf{q}_k^\dag \mathbf{q}_k} \sqrt{c_k M}} - \left(1+\frac{1}{c_k} \right) \right]^+}.
\end{align}

\subsubsection{Estimation of $g_k$} 
The MMSE estimator for $g_k$, i.e. $\hat{g}_k = \mathbb{E} \left[ g_k | \mathbf{r}_k , \mathbf{q}_k , \zeta_k ; w_k \right] $ can be derived as 
\begin{equation} \label{eq:gk}
	\hat{g}_{k} = \frac{\mathbf{r}_k^\dag \mathbf{q}_k + \sigma^2_{n_k} \mu_{g_k}/ \sigma^2_{g_k} }{\mathbf{q}_k^\dag \mathbf{q}_k + \sigma^2_{n_k}/ \sigma^2_{g_k} }.
\end{equation}
The details of the derivation are given in Appendix \ref{App:MMSE_gk}.

\subsubsection{Estimation of $g^e_k$}
The MMSE estimator for $g^e_k$ is obtained by taking the conditional expectation of $g^e_k$ given the known parameters, $\mathbf{r}_k$, $ \mathbf{q}_k$, and $\zeta_k$  \cite{Poor_1994}. That is, $\hat{g}^e_k$ is given by
 \begin{align}\label{eq:MMSE_gek}
 	\hat{g}^e_k
	&= \mathbb{E}[ g^e_k | \mathbf{r}_k , \mathbf{q}_k , \zeta_k ; w_k],
\end{align}
which is derived in the following theorem. 
\begin{thm}\label{thm:MMSE_gek}
For given $\mathbf{r}_k$, $ \mathbf{q}_k$, and $\zeta_k$, the MMSE 
estimator for $g^e_k$ is given by 
\begin{align}\label{eq:MMSE_gek2}
	\hat{g}^e_k = \frac{w_k c_k}{1+ w^2_k c_k} \left(\frac{\zeta_k}{\sqrt{c_k M}}- \hat{g}_k \right),
\end{align}
where $\hat{g}_k$ is the MMSE estimate of $g_k$ in \eqref{eq:gk}. 
\end{thm}
\begin{proof}
From \eqref{eq:MMSE_gek}, we can find $\hat{g}^e_k$ by conducting a serious of decomposition as follows: 
\ifCLASSOPTIONonecolumn 
\begin{align}\label{eq:MMSE_gek_proof}
	\hat{g}^e_k
	&= \int  g^e_k f(g^e_k | \mathbf{r}_k , \mathbf{q}_k , \zeta_k ; w_k) d g^e_k  
	 = \int  g^e_k \int f(g^e_k, g_k | \mathbf{r}_k , \mathbf{q}_k , \zeta_k; w_k) d g_k d g^e_k  \nonumber\\
	&= \int  g^e_k \int f(g^e_k | g_k , \mathbf{r}_k , \mathbf{q}_k , \zeta_k; w_k) f(g_k | \mathbf{r}_k , \mathbf{q}_k , \zeta_k; w_k)d g_k d g^e_k  \nonumber\\
	&\mathop = \limits^{(a)}  \int  g^e_k \int f(g^e_k | g_k , \zeta_k; w_k) 
	f(g_k | \mathbf{r}_k , \mathbf{q}_k , \zeta_k; w_k)d g_k d g^e_k  \nonumber \\
	&\mathop = \limits^{(b)}  \int  \left[ \int g^e_k f(g^e_k | g_k , \zeta_k; w_k) dg^e_k  \right] 
	f(g_k | \mathbf{r}_k , \mathbf{q}_k , \zeta_k; w_k)d g_k   \nonumber \\
	&\mathop = \limits^{(c)}   \frac{w_k c_k}{1+ w_k^2 c_k} \left( \frac{\zeta_k}{\sqrt{c_k M}} - \mathbb{E} \left[ g_k | \mathbf{r}_k , \mathbf{q}_k , \zeta_k; w_k \right]   \right) ,
 \end{align}
 \else 
 \begin{align}\label{eq:MMSE_gek_proof}
	\hat{g}^e_k
	&= \int  g^e_k f(g^e_k | \mathbf{r}_k , \mathbf{q}_k , \zeta_k ; w_k) d g^e_k  \nonumber\\
	&= \int  g^e_k \int f(g^e_k, g_k | \mathbf{r}_k , \mathbf{q}_k , \zeta_k ; w_k) d g_k d g^e_k  \nonumber\\
	&= \int  g^e_k \int f(g^e_k | g_k , \mathbf{r}_k , \mathbf{q}_k , \zeta_k ; w_k) \nonumber\\
	&\qquad\qquad\qquad\qquad\qquad\quad f(g_k | \mathbf{r}_k , \mathbf{q}_k , \zeta_k ; w_k)d g_k d g^e_k  \nonumber\\
	&\mathop = \limits^{(a)}  \int  g^e_k \int f(g^e_k | g_k , \zeta_k ; w_k) 
	    f(g_k | \mathbf{r}_k , \mathbf{q}_k , \zeta_k ; w_k)d g_k d g^e_k  \nonumber \\
	&\mathop = \limits^{(b)}  \int  \left[ \int g^e_k f(g^e_k | g_k , \zeta_k ; w_k) dg^e_k  \right] 
	f(g_k | \mathbf{r}_k , \mathbf{q}_k , \zeta_k ; w_k)d g_k   \nonumber \\
	&\mathop = \limits^{(c)}   \frac{w_k c_k}{1+ w_k^2 c_k} \left( \frac{\zeta_k}{\sqrt{c_k M}} - \mathbb{E} \left[ g_k | \mathbf{r}_k , \mathbf{q}_k , \zeta_k ; w_k \right]   \right) ,
 \end{align}
 \fi
 where $(a)$ is from the fact that $g^e_k$ is independent of $ \mathbf{r}_k $ and $ \mathbf{q}_k $ for given $g_k $ and $ \zeta_k $, $(b)$ is from the Fubini's theorem that allows us to change the order of integration, $(c)$ is from the fact that $\int g^e_k f (g^e_k | g_k , \zeta_k ; w_k ) dg^e_k =\mu_{c} $ derived in Appendix \ref{App:pdfs}.
\end{proof}
\begin{rmk}
The expression for $\hat{g}^e_k$ in (\ref{eq:MMSE_gek2}) shows that the estimator exploits the complementary relation as expected. 
That is, Bob first tries to estimate $g_k$, and then $g^e_k$ 
by comparing the estimate of
$g_k$ with its expected value, $\frac{\zeta_k}{\sqrt{c_k M}}$.
Thus, for a given $\zeta_k$, the smaller $\hat{g}_{k}$ Bob has, the larger $\hat{g}^e_k$ the estimator produces. In the end, the PCA is detected with the higher probability. It should be noted that while we do not discuss in this work, a detector can be derived with the result in \eqref{eq:MMSE_gek2} with which the shared random sequence can be discarded when eavesdropping is detected as the BB84 protocol does. In the evaluations of the $\hat{g}^e_k$, the unknown parameter $w_k$ will be replaced with the estimated one, $\hat{w}_k$ in \eqref{eq:MLE_wk}. 
\end{rmk}

In Section \ref{Sec:Results}, 
the performances of the proposed estimator in \eqref{eq:MMSE_gek2} will 
be  evaluated in terms of the MSE \cite{Poor_1994} which is derived 
in this subsection.  Let us first rewrite the estimate of $g_k$ as $\hat{g}_k = g_k - e_{k,1}$, where $e_{k,1}$ is the MMSE estimation error 
of $\hat{g}_k$.  Then, the MMSE 
estimator for $g^e_k$ in \eqref{eq:MMSE_gek2} can be rewritten as 
\begin{align}\label{Structure of MMSE estimator of ge}
	\hat{g}^e_k
	&= \frac{w_k c_k}{1+ w^2_k c_k} \left(\frac{\zeta_k}{\sqrt{c_k M}}- g_k + e_{k,1} \right) \nonumber\\
	&\mathop = \limits^{(a)} \mathbb{E}[g^e_k|\zeta_k ,g_k] + \frac{w_k c_k}{1+w^2_k c_k} e_{k,1} \nonumber\\
	& = g^e_k - e_{k,2} + \frac{w_k c_k}{1+w^2_k c_k} e_{k,1} 
	 = g^e_k - e_{k,2} + e^\prime_{k,1}, 
\end{align}
where $(a)$ is due to the fact that the conditional mean value of $g^e_k$ for given $g_k$ and $\zeta_k$ is $\frac{w_k c_k}{1+ w^2_k c_k} (\frac{\zeta_k}{\sqrt{c_k M}}- g_k)$ as shown in Appendix \ref{App:pdfs}, and $e_{k,2}$ is defined as the MMSE estimation error of $g^e_k$ when $g_k$ is perfectly known to Bob.
Thus, the MSE of $\hat{g}^e_k$ becomes
\ifCLASSOPTIONonecolumn 
\begin{align*}
	\mathbb{E}[ |\hat{g}^e_k - g^e_k|^2; w_k ] 
	   = \, & \mathbb{E}[| e_{k,2} |^2] + \mathbb{E} [| e^\prime_{k,1} |^2 ] - 2 \mathbb{E} [ \Re \{e_{k,2} e^\prime_{k,1}\}] 
	 \mathop = \limits^{(a)} \, \mathbb{E} [| e_{k,2} |^2 ] + \left(\frac{w_k c_k}{1+w^2_k c_k}\right)^2 \mathbb{E} [ | e_{k,1} |^2 ],
\end{align*}
\else
\begin{align*}
	\mathbb{E}[ |\hat{g}^e_k & -  g^e_k|^2; w_k ] \\
	& = \mathbb{E}[| e_{k,2} |^2] + \mathbb{E} [| e^\prime_{k,1} |^2 ] - 2 \mathbb{E} [ \Re \{e_{k,2} e^\prime_{k,1}\}] \nonumber\\
	& \mathop = \limits^{(a)} \mathbb{E} [| e_{k,2} |^2 ] + \left(\frac{w_k c_k}{1+w^2_k c_k}\right)^2 \mathbb{E} [ | e_{k,1} |^2 ],
\end{align*}
\fi
where $(a)$ is from the facts that $e_{k,2}$ is independent of $e_{k,1}$, and $\mathbb{E}[e_{k,i}]=0$ for $i=1,2$ since the MMSE estimation error is Gaussian with zero mean  \cite{Poor_1994, ChoiBook}. 
Note that $\mathbb{E} [ | e_{k,2} |^2 ]$ and $\mathbb{E}[| e_{k,1} |^2 ]$ correspond to the conditional variances of $g^e_k$ and $g_k$ with respect to the pdfs $f(g^e_k| \zeta_k, g_k ; w_k)$ and $f (g_k | \mathbf{r}_k , \mathbf{q}_k , \zeta_k ; w_k)$, respectively, which are derived in Appendices \ref{App:pdfs} and \ref{App:MMSE_gk}, respectively. Then, the MSE of $\hat{g}^e_k$ is finally given by 
\ifCLASSOPTIONonecolumn 
\begin{align} \label{eq:MSE_gek}
	\mathbb{E} [| \hat{g}^e_k - g^e_k|^2; w_k]
	& = \frac{1}{(1+w^2_k c_k)M} + \left(\frac{w_k c_k}{1+w^2_k c_k}\right)^2 \frac{\sigma^2_{g_k}\sigma^2_{n_k}}{\sigma^2_{g_k} \mathbf{q}^\dag_k \mathbf{q}_k+\sigma^2_{n_k}} \nonumber\\
	& = \frac{1}{M} \left\{	\frac{1}{1+w^2_k c_k} + 
		  \frac{ \left(\frac{w_k c_k}{1+w^2_k c_k}\right)^2}
			   { \left(K-1+\frac{1}{p_d\beta_k}\right) ^{-1} \mathbf{q}^\dag_k \mathbf{q}_k + \frac{1+ c_k + w^2_k c_k}{1+w^2_k c_k} } \right\}.
\end{align}
\else
\begin{align} \label{eq:MSE_gek}
	&\mathbb{E} [| \hat{g}^e_k - g^e_k|^2; w_k] \nonumber \\
	& = \frac{1}{(1+w^2_k c_k)M} + \left(\frac{w_k c_k}{1+w^2_k c_k}\right)^2 \frac{\sigma^2_{g_k}\sigma^2_{n_k}}{\sigma^2_{g_k} \mathbf{q}^\dag_k \mathbf{q}_k+\sigma^2_{n_k}} \nonumber\\
	& = \frac{1}{M} \left\{	\frac{1}{1+w^2_k c_k} + 
		  \frac{ \left(\frac{w_k c_k}{1+w^2_k c_k}\right)^2}
			   { \frac{\mathbf{q}^\dag_k \mathbf{q}_k}{\left(K-1+\frac{1}{p_d\beta_k}\right)} + \frac{1+ (1 + w^2_k)c_k}{1+w^2_k c_k} } \right\},
\end{align}
\fi
where $\sigma^2_{n_k}$ is taken from \eqref{EQ:sn}. The MSE in (\ref{eq:MSE_gek}) looks inversely proportional to the number of antennas, $M$.

\begin{rmk}
  Before closing this sub-section, it should be noted that Alice can also transmit downlink pilot signals which may help Bob to estimate the EDCGs, $g_k$, and $g^e_k$. Such a two-way training strategy \cite{Gomadam08Channel,Jose11Channel} is especially helpful under fast fading environments where the length of random sequence $\mathbf{q}_k$ is limited due to the short coherence time interval. Since the proposed scheme estimates the EDCGs based on the random sequence $\mathbf{q}_k$, the limited length of random sequence results in poor estimates of $g_k$ and $g^e_k$. In such a case, it may be beneficial to allocate a portion of the downlink transmission to the downlink pilot signals, which significantly improves the estimates. While the two-way training strategy reduces the length of random sequence, the better estimates of EDCGs may offset the decrease of the length of the random sequence. Thus, by carefully allocating the downlink transmission to the pilot signals and random sequence, a longer secret key may be achievable. As another practical issue, we may need to consider that while this work assumes the perfect channel reciprocity, the uplink and downlink channels may change during the key sharing process, which increases the channel estimate error. The analysis with the channel variation can be conducted by modifying the derivations for the estimation errors derived in this section.
\end{rmk}

\subsection{Secret Key Generation}
We can now estimate $I (Q_k ; R^e_k)$  in \eqref{eq:Key1} by replacing the true value of $g^e_k$ with its estimate, $\hat{g}^e_k$ derived in the previous 
subsection. Let us denote the estimate of $I (Q_k ; R^e_k)$ based on $\hat{g}^e_k$ by $I(Q_k ; \hat{R}^e_k)$ with which the length of secret key is now adaptively determined as  
\begin{align}\label{eq:Key3}
	\hat{s}_k & = [ N_d \{  I(Q_k ; R_k) - I ( Q_k ; \hat{R}^e_k) \} -2a_k -2 - b_k ]^+ \\
	          & = [ N_d \{  H(Q_k|\hat R^e_k) - H(Q_k|R_k) \} -2a_k -2 - b_k ]^+ \nonumber \\
	          & \ge [ N_d \{  H(Q_k|\hat R^e_k) - \nu_k \} -2a_k -2 - b_k ]^+, \label{eq:Key31}
\end{align}
where $\hat R^e_k = \hat g^e_k Q_k + N_k$ is the output from an AWGN channel with gain $\hat g^e_k$ and Gaussian noise $N_k$, and $\nu_k \ge H(Q_k|R_k)$ is the number of parity bits exchanged in the information reconciliation phase. Later in the performance evaluations, we assume that $\nu_k$ equals 
$H(Q_k|R_k)$, and the equality holds in \eqref{eq:Key31}.
However, the estimation error in $\hat{g}^e_k$ may result in an underestimate 
of $I (Q_k ; {R}^e_k)$.
According to \eqref{eq:Key2}, a secrecy outage occurs when the true value of $g^e_k$ is greater than its estimate, $\hat{g}^e_k$. In this case, Alice and Bob fail to make the generated secret key completely secure from the eavesdropping. To analyze how often the proposed SKA protocol causes the outage event, we evaluate the secrecy outage probability, i.e $\Pr(|\hat{g}^e_k| < |g^e_k|)$.

We introduce to the SKA protocol a design parameter called a secrecy margin, $\delta \in [0, \infty)$ to compensate for the estimation error in $\hat{g}^e_k$ and define an outage event as
\begin{align}\label{eq:out_event}
	\mathcal{S}=\{| g^e_k |^2 \,|\,(1 + \delta)^2 |\hat{g}^e_k|^2 < |g^e_k|^2\}.
\end{align}
Then, for given $w_k$ and $\delta$, the secrecy outage probability is expressed as
\ifCLASSOPTIONonecolumn 
\begin{align}\label{OPforGivenCB}
	P_{\rm out}(\mathbf{r}_k,\mathbf{q}_k, \zeta_k; w_k, \delta) 
	 = \int_{|g^e_k|^2\in\mathcal{S}}f(|g^e_k|^2|\mathbf{r}_k,\mathbf{q}_k, \zeta_k; w_k)d|g^e_k|^2.
\end{align}
\else
\begin{align}\label{OPforGivenCB}
	&P_{\rm out}(\mathbf{r}_k,\mathbf{q}_k, \zeta_k; w_k, \delta) \nonumber\\
	& = \int_{|g^e_k|^2\in\mathcal{S}}f(|g^e_k|^2|\mathbf{r}_k,\mathbf{q}_k, \zeta_k;w)d|g^e_k|^2.
\end{align}
\fi
Meanwhile, the secrecy margin turns $\hat R^e_k$ into $\hat R^e_k = (1 + \delta) \hat g^e_k Q_k + N_k$, which decreases the secret key length, $\hat s_k$ in \eqref{eq:Key3} at the expense of the outage probability.

Since the pdf $f(g^e_k|\mathbf{r}_k,\mathbf{q}_k, \zeta_k; w_k)$ follows $\mathcal{CN}(\mu_{\hat{g}^e_k},\sigma^2_{\hat{g}^e_k})$ as shown in Appendix \ref{App:pdfs}, the conditional pdf $f( |g^e_k|^2 |\mathbf{r}_k,\mathbf{q}_k, \zeta_k; w_k)$ in \eqref{OPforGivenCB} follows a Rice distribution. 
Thus, we can simply express $P_{\rm out}(\mathbf{r}_k,\mathbf{q}_k, \zeta_k; w_k, \delta)$ in \eqref{OPforGivenCB} as 
\ifCLASSOPTIONonecolumn 
\begin{align}\label{OPforGivenCBSimpler}
		P_{\rm out}(\mathbf{r}_k,\mathbf{q}_k, \zeta_k; w_k, \delta)
		=Q_1\left(\sqrt{\frac{2|\mu_{\hat{g}^e_k}|^2}{\sigma^2_{\hat{g}^e_k}}},
						 \sqrt{\frac{2|(1+\delta)\hat{g}^e_k|^2}{\sigma^2_{\hat{g}^e_k}}}\right),
\end{align}
\else
\begin{align}\label{OPforGivenCBSimpler}
		&P_{\rm out}(\mathbf{r}_k,\mathbf{q}_k, \zeta_k; w_k, \delta) \nonumber\\
		&=Q_1\left(\sqrt{\frac{2|\mu_{\hat{g}^e_k}|^2}{\sigma^2_{\hat{g}^e_k}}},
						 \sqrt{\frac{2|(1+\delta)\hat{g}^e_k|^2}{\sigma^2_{\hat{g}^e_k}}}\right),
\end{align}
\fi
where $Q_1(a,b)=\int^\infty_b x\exp\left(-\frac{x^2+a^2}{2}\right) I_0(ax)dx$ is the first order generalized Marcum $Q$-function, and $I_0 (\cdot)$ is the zero-th order modified Bessel function. By averaging the secrecy outage probability in \eqref{OPforGivenCBSimpler} with respect to the joint pdf $f( \mathbf{r}_k , \mathbf{q}_k , \zeta_k ;w_k)$, the average secrecy outage probability is given by 
\begin{align}\label{AvgOP} 
	\bar{P}_{\rm out}(w_k,\delta) 
	& = \mathbb{E} [P_{\rm out}(\mathbf{r}_k,\mathbf{q}_k, \zeta_k; w_k, \delta)].
\end{align}

To get more insightful results, we introduce to \eqref{OPforGivenCBSimpler} an upper bound 
\begin{align}\label{UB1ofOP}
	Q_1(a,b)&\leq\exp\left[-\frac{(b-a)^2}{2}\right],
\end{align}
for $a\leq b$ \cite{Simon00Exponential,Corazza_PoutUB}. Then, the average secrecy outage probability is also bounded as
\ifCLASSOPTIONonecolumn 
\begin{align}\label{UBOnOP} 
    \bar{P}_{\rm out}(w_k,\delta) 
	&\leq \mathbb{E} \left[ \exp\left\{-\frac{\left(| (1+\delta)\hat{g}^e_k| - | \mu_{\hat{g}^e_k} | \right)^2}{\sigma^2_{\hat{g}^e_k}}\right\}\right]. 
\end{align}
\else
\begin{align}\label{UBOnOP} 
	& \bar{P}_{\rm out}(w_k,\delta) 
	 \leq \mathbb{E} \left[ \exp\left\{-\frac{\left(| (1+\delta)\hat{g}^e_k| - | \mu_{\hat{g}^e_k} | \right)^2}{\sigma^2_{\hat{g}^e_k}}\right\} \right].
\end{align}
\fi
Now, it is much easier to investigate into asymptotic behaviors of the average secrecy outage probability with the upper bound in \eqref{UBOnOP}. That is, as $M$ grows, the upper bounds in \eqref{UBOnOP} turns into  
\ifCLASSOPTIONonecolumn
\begin{equation}\label{UBOnOP2}
  \lim_{M \rightarrow \infty}   \bar{P}_{\rm out}(w_k,\delta) 
     \le \exp\left[-\frac{(1+w^2_k c_k)M\delta^2}{1 + \frac{w^2_k c^2_k \left(K-1+\frac{1}{p_d\beta_k}\right)}{\mathbf{q}^\dag_k \mathbf{q}_k(1+w^2_k c_k) + \left(K-1+\frac{1}{p_d\beta_k}\right)\{1 + (1+w^2_k)c_k\}}}\right].
\end{equation}
\else
\begin{align}
  \lim_{M \rightarrow \infty} & \bar{P}_{\rm out}(w_k,\delta) \label{UBOnOP2}\\
     \le & \exp\left[-\frac{(1+w^2_k c_k)M\delta^2}{1 + \frac{w^2_k c^2_k \left(K-1+\frac{1}{p_d\beta_k}\right)}{\mathbf{q}^\dag_k \mathbf{q}_k(1+w^2_k c_k) + \left(K-1+\frac{1}{p_d\beta_k}\right)\{1 + (1+w^2_k)c_k\}}}\right]. \nonumber
\end{align}
\fi
since $\hat{g}^e_k \to  \mu_{\hat{g}^e_k}$ in probability.

It can be noticed that the bound in \eqref{UBOnOP2} decreases exponentially fast with $M$, which is possible due to the fact that Bob can accurately adjust the length of extracting secret key as $\hat{g}^e_k$ gets close to $g^e_k$ with the growing number antennas. That is, even if the beam is tilted toward Eve under the PCA, the LAA is still helpful for Bob not only to eliminate the noise/fading effects from the received signal but also to estimate the amount of information leakage. In Section \ref{Sec:Results}, we will analyze the average secrecy outage  probability with various combinations of $w_k$ and $\delta$. 

\section{Performance Evaluations}\label{Sec:Results}
In this section, we present performances of the proposed SKA protocol in both numerical and analytic ways\footnote{Although we also perform  Monte Carlo simulations, we do not present the results in this paper as we confirm that our simulation results are exactly the same to the numerical evaluations.}. Throughout the performance evaluations, we consider the following system setup: the ratios of uplink and downlink transmit powers ($p_u$ and $p_d$, respectively) to the unit noise variance are set
to 10 dB and 20 dB, respectively. This asymmetric power allocation is due to the practical consideration that the power at the user side may be limited. For the downlink transmission, we consider binary random sequences, i.e $\mathbf{q}_k \in \{-1, 1\}^{N_d}$ for $k \in \mathcal{K}$. While binary sequences are considered in this work, the results can be readily extended to sequences with larger alphabet symbols since all the derivations are in general forms. The large scale fading factors, i.e. $\beta_k$ and $\beta^e_\ell$, are set to one for all $k, \ell \in \mathcal{K}$. Finally, the number of symbols in the uplink training sequences is assumed to be $N_u =  100$. For the given setup, we will evaluate performances of the proposed SKA protocol with various combinations of design parameters, $M$, $K$, $N_d$, and $\delta$ and show how the parameters affect the performances.
\subsection{Estimation of $g^e_k$}\label{Estimation of PCA Parameters}
\ifCLASSOPTIONonecolumn
\begin{figure*}[h]
	\centering
		\subfigure[NMSE of $\hat{g}^e_k$ versus $M$ when ${N_d = 1,000}$.]
		{
			\includegraphics[width=\SubFigWidth]{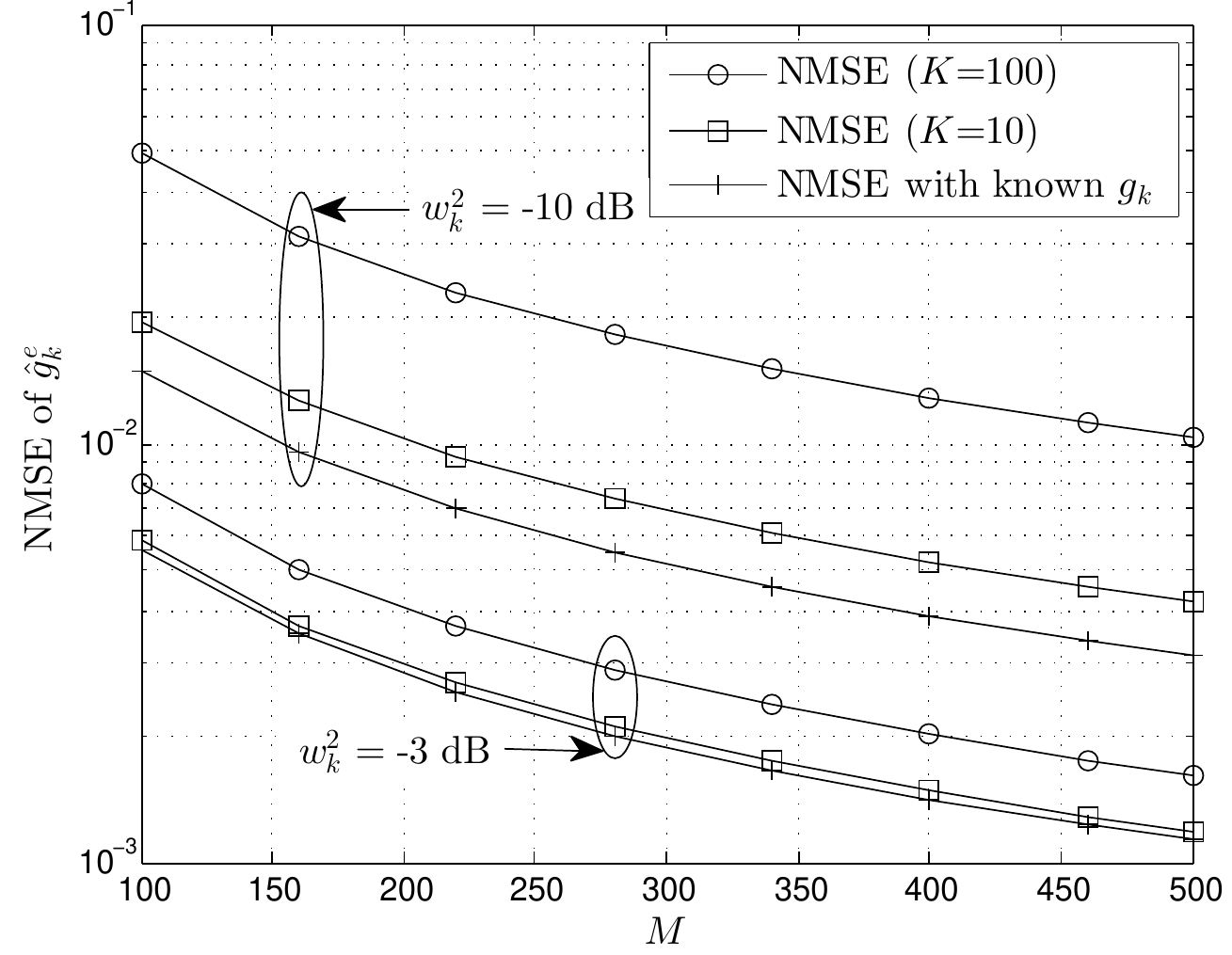}
			\label{Fig:MSE1}
		}
		\subfigure[NMSE of $\hat{g}^e_k$ versus $N_d$  when ${M=500}$.]
		{
			\includegraphics[width=\SubFigWidth]{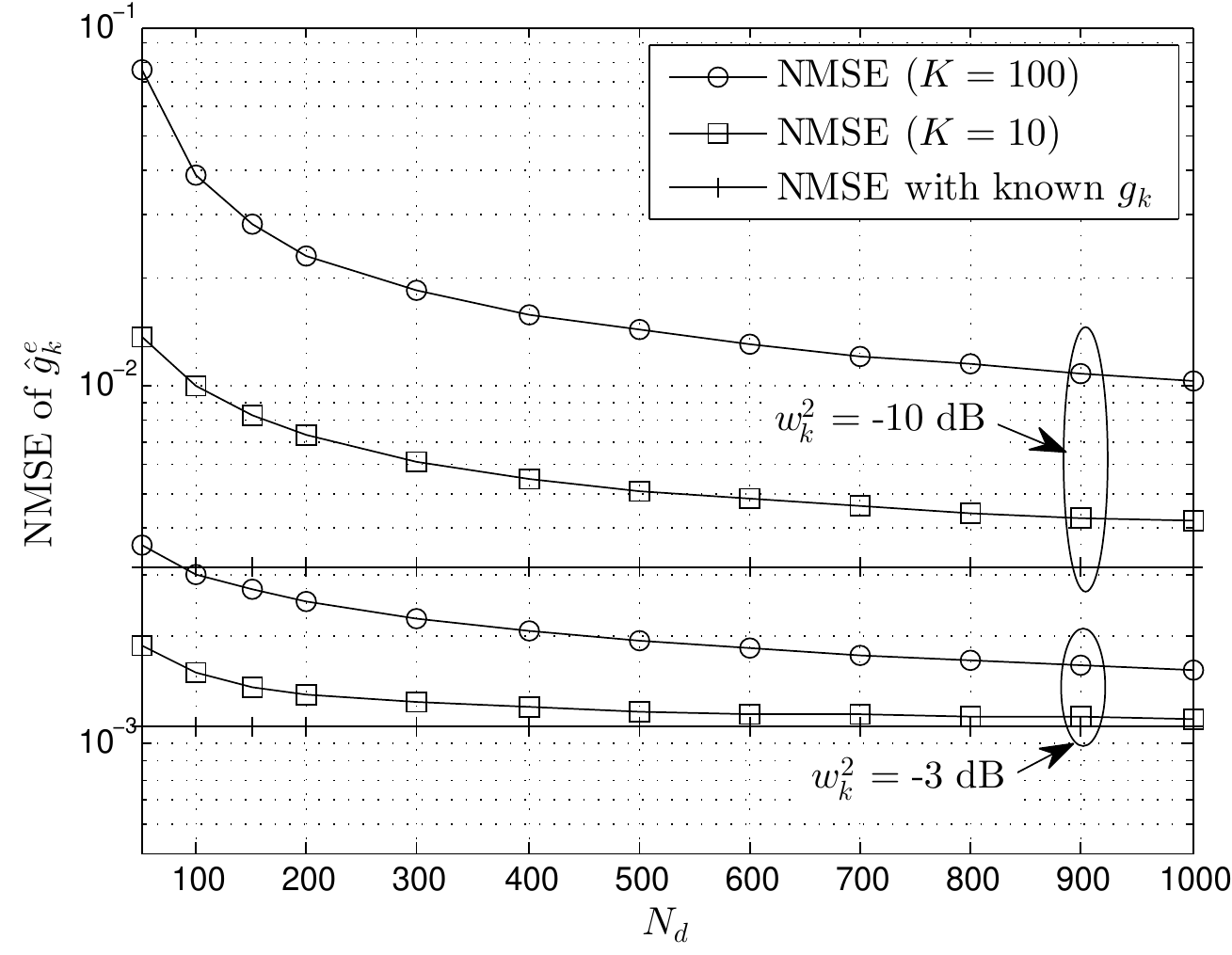}
			\label{Fig:MSE2}
		}
	\caption{NMSE of $\hat{g}^e_k$ with different values of $M$, $K$, and $N_d$. }
	\label{Fg:Fig2}
\end{figure*}
\else
\begin{figure*}[t]
	\centering
		\subfigure[NMSE of $\hat{g}^e_k$ versus $M$ when ${N_d = 1,000}$.]
		{
			\includegraphics[width=\SubFigWidth]{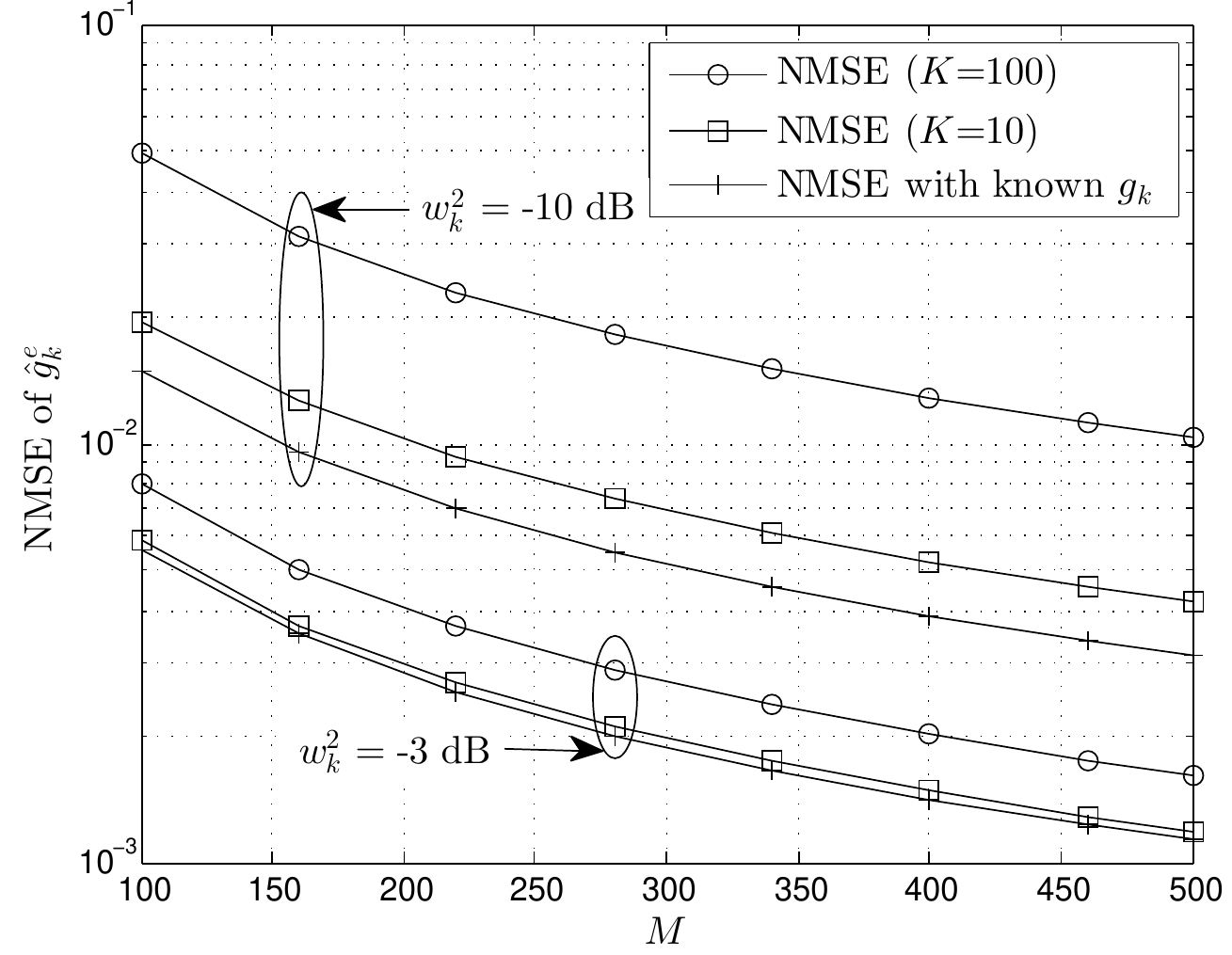}
			\label{Fig:MSE1}
		}
		\subfigure[NMSE of $\hat{g}^e_k$ versus $N_d$  when ${M=500}$.]
		{
			\includegraphics[width=\SubFigWidth]{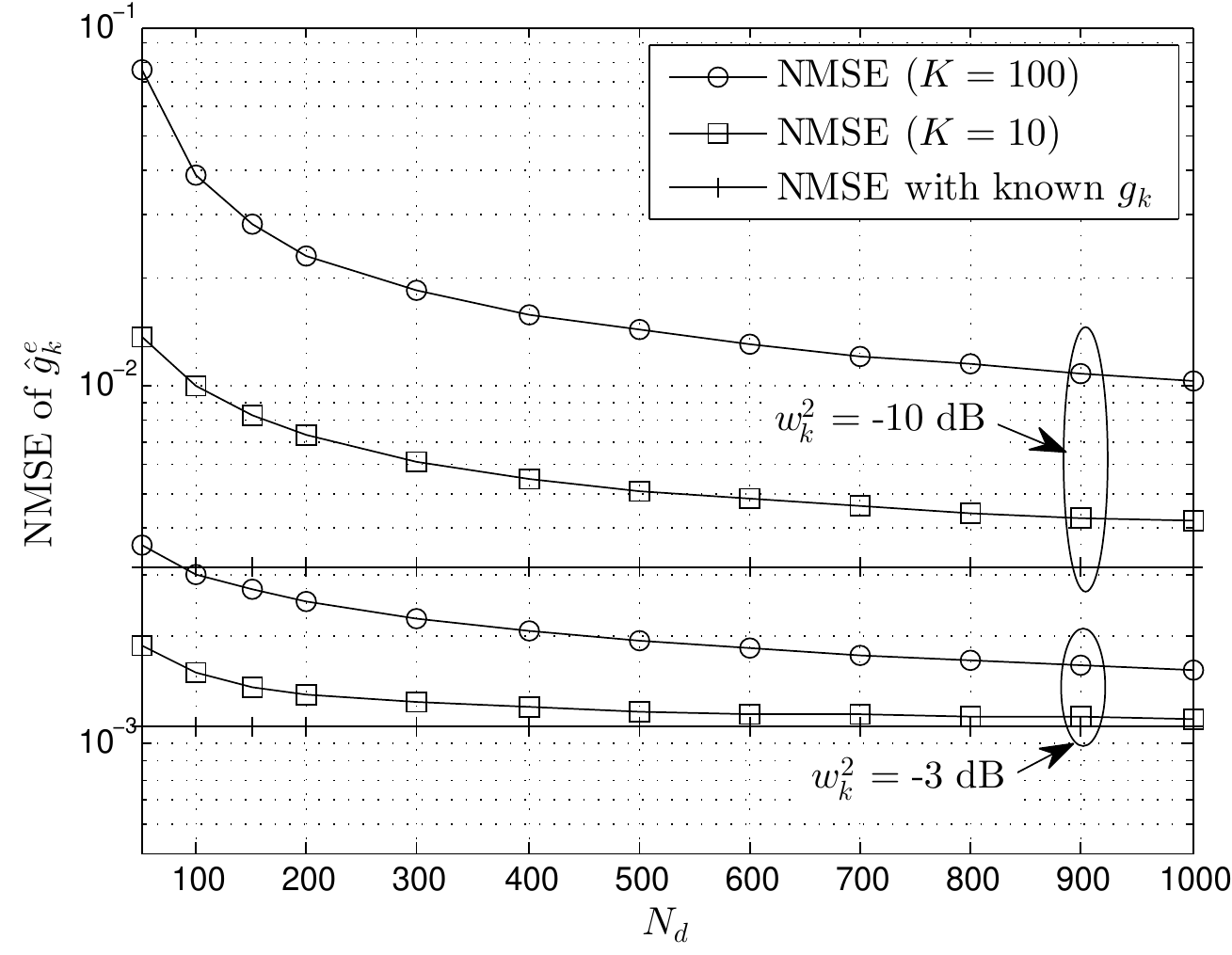}
			\label{Fig:MSE2}
		}
	\caption{NMSE of $\hat{g}^e_k$ with different values of $M$, $K$, and $N_d$. }
	\label{Fg:Fig2}
\end{figure*}
\fi
We first evaluate performances of the estimator $\hat{g}^e_k$ in \eqref{eq:MMSE_gek2} in terms of the normalized MSE (NMSE) defined as
\begin{equation}\label{eq:NMSE}
	\text{NMSE} = \frac{\mathbb{E}[(|g^e_k - \hat{g}^e_k|)^2]}{\mathbb{E}[|g^e_k|^2]}.
\end{equation}
As a reference, we consider the case that $g_k$ is perfectly known to Bob. Then, the NMSE becomes
\begin{equation} \label{eq:LBNMSE}
{\rm NMSE}_{\rm ideal}
=
 \frac{1}{\mathbb{E}[|g^e_k|^2]} \frac{1}{(1+ w_k^2 c_k)M},
\end{equation}
since the perfect knowledge of $g_k$ makes $e_{k,1}$ in \eqref{Structure of MMSE estimator of ge} zero.  Thus, the 
ideal NMSE in \eqref{eq:LBNMSE} is a lower bound on the NMSE of $\hat{g}^e_k$ which is taken as a yardstick in the performance evaluations.  The NMSE evaluations are carried out with the analytic expression of the MSE in \eqref{eq:MSE_gek} and the estimate of $w_k$ in \eqref{eq:MLE_wk}. To confirm the results, we also find an empirical expectation for the NMSE by conducting the estimations of $g^e_k$ and $w_k$ $10^5$ times at each point in Figs. \ref{Fig:MSE1} and \ref{Fig:MSE2}. The evaluations from the two different approaches are completely overlapped each other, which confirms the derivations in Section \ref{Sec:EstInfLeak}.

The evaluations of NMSE versus $M$ at different combinations of $K$ and $w_k$ are depicted in Fig. \ref{Fig:MSE1} where the NMSE decreases as either of $M$ and $w_k$ increases as predicted by the closed-form expression 
for the MSE in \eqref{eq:MSE_gek}. The results in Fig \ref{Fig:MSE1} imply that for a fixed number of antennas, $M$, Bob achieves a better estimation of Eve's channel $g^e_k$ when Eve attempts a stronger attack, i.e. a larger PCA power $w^2_k$ in hopes of eavesdropping more information about the communication between legitimate parties. On the other hand, the increase
of the number of users, $K$, induces more interference to Bob and thus degrades the NMSE. In Fig. \ref{Fig:MSE2}, the NMSE evaluations are performed with respect to $N_d$ at $M=500$, which shows the same trends as the ones in Fig. \ref{Fig:MSE1}. It is also noticed that the lower bound on NMSE looks independent of $N_d$, i.e. the length of $\mathbf{q}_k$ since $N_d$ affects only the results of $\hat{g}_k$ as shown in \eqref{eq:gk}, and the lower bound already assumes the true value of $g_k$. The results in Fig. \ref{Fig:MSE2} show that the performances of $\hat{g}^e_k$ eventually approaches the lower bounds as $N_d$ increases since the more samples provide the better degree of accuracy in the estimation of $g_k$\cite{Poor_1994, ChoiBook}.

\subsection{Average Secrecy Outage Probability}\label{Outage Probability of Secrecy Rate}
\begin{figure*}[t]
	\centering
		\subfigure[$\bar{P}_{\rm out}(w_k,\delta) $ versus $M$ when $N_d = 1,000$.]
		{
			\includegraphics[width=\SubFigWidth]{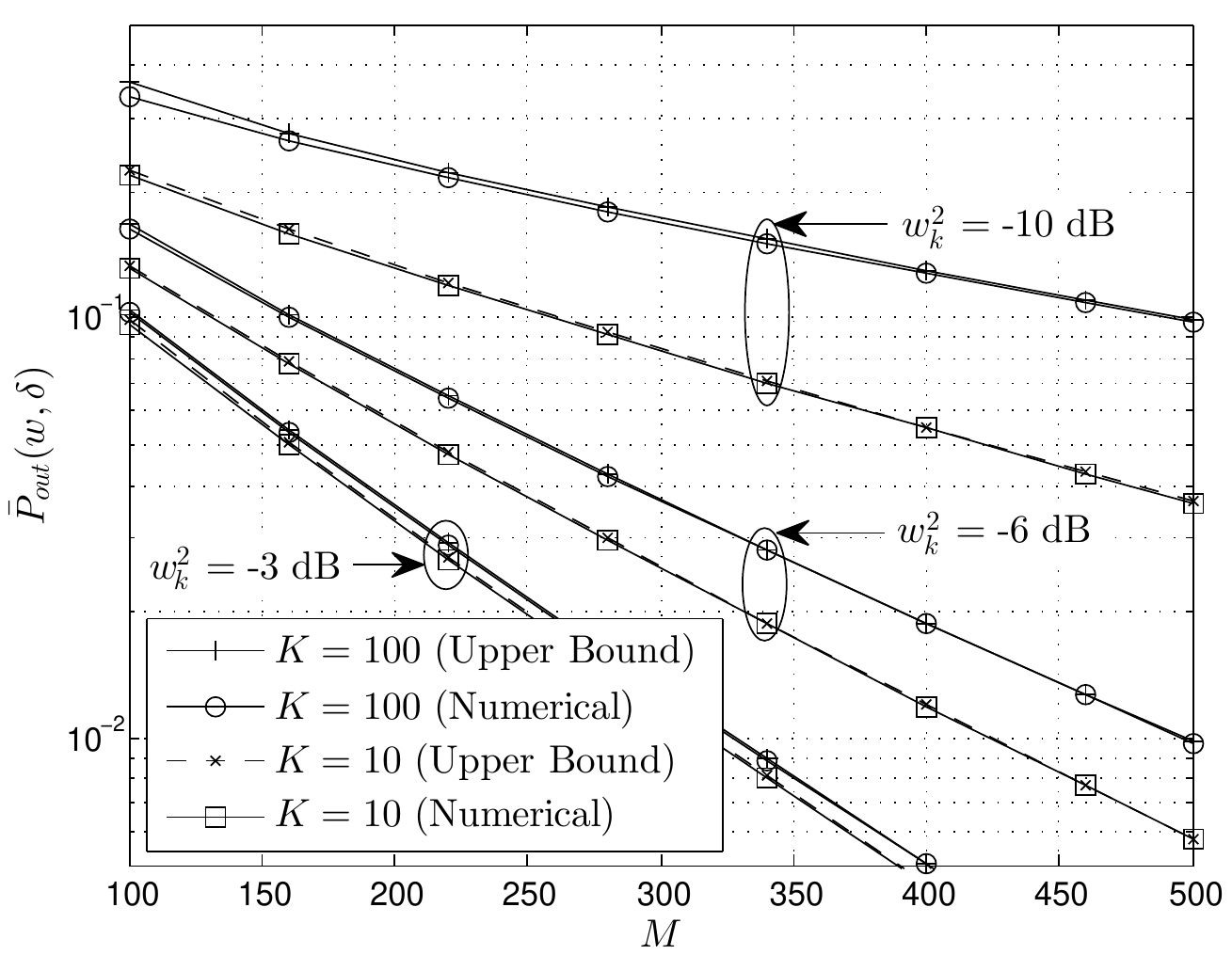}
			\label{Fig:OutageProb1}
		}
		\subfigure[$\bar{P}_{\rm out}(w_k,\delta) $ versus $N_d$ when $M=500$.]
		{
			\includegraphics[width=\SubFigWidth]{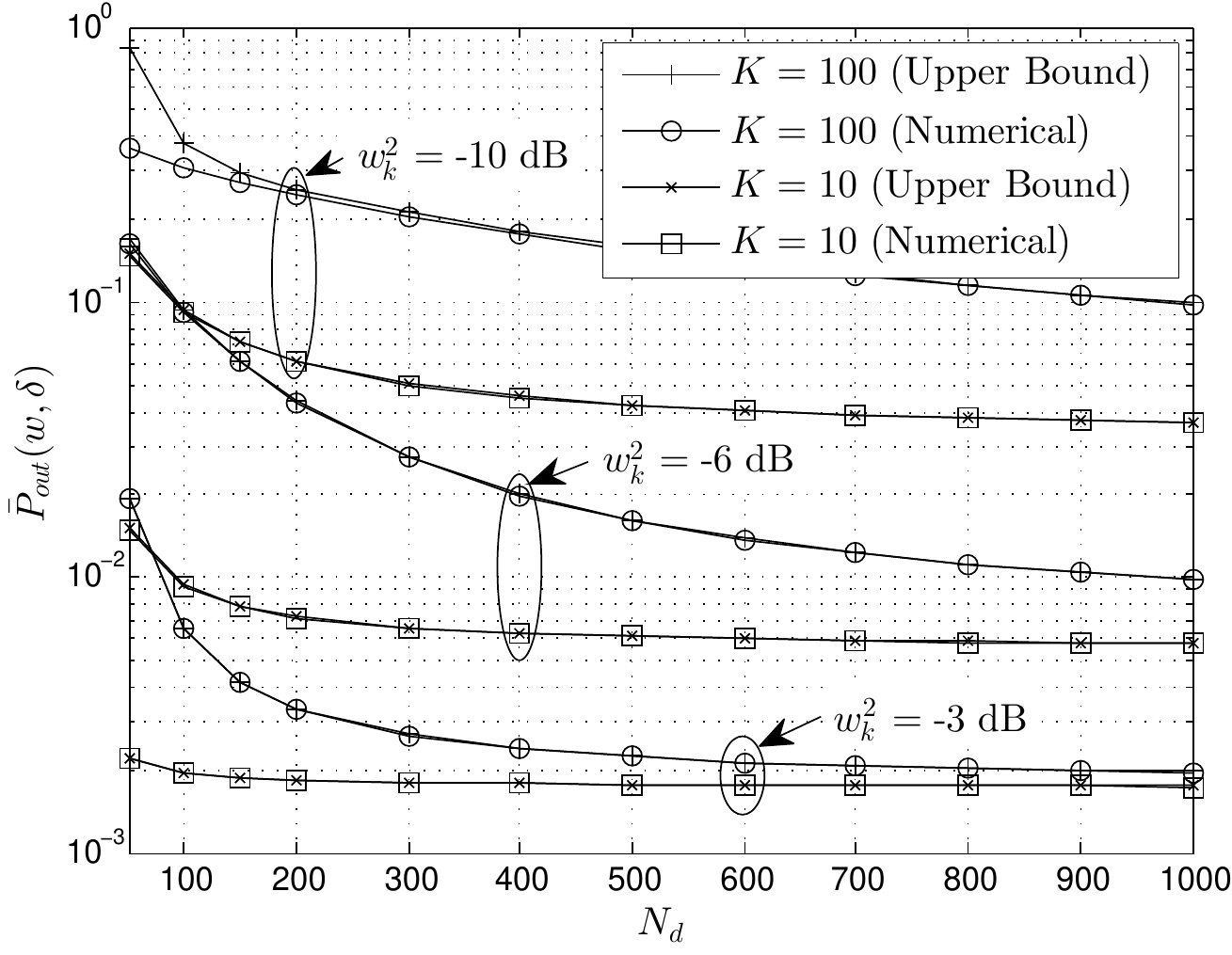}
			\label{Fig:OutageProb3}
		}
	\caption{Average secrecy outage probability, $\bar{P}_{\rm out}(w_k,\delta) $, with different parameters for a fixed  $\delta=0.1$.}
	\label{Fg:OutageProb}
\end{figure*}
In this subsection, we present average secrecy outage probability of the proposed SKA protocol. As a performance benchmark, we also plot the upper bound 
on $\bar{P}_{\rm out}(w_k,\delta) $ by introducing another upper bound on the Marcum $Q$-function: \eqref{OPforGivenCBSimpler}:
\ifCLASSOPTIONonecolumn
\begin{align}\label{UB2ofOP}
	Q_1(a,b)&\leq \frac{I_0(ab)}{\exp(ab)}
			  \left\{\exp\left[-\frac{(b-a)^2}{2}\right] + 
			  		 a\sqrt{\frac{\pi}{2}}\text{erfc}\left(\frac{b-a}{\sqrt{2}}\right)							  \right\},
\end{align}
\else
\begin{align}\label{UB2ofOP}
	&Q_1(a,b) \nonumber \\
	&\leq \frac{I_0(ab)}{\exp(ab)} 
			  \left\{\exp\left[-\frac{(b-a)^2}{2}\right] + 
			  		 a\sqrt{\frac{\pi}{2}}\text{erfc}\left(\frac{b-a}{\sqrt{2}}\right)							  \right\},
\end{align}	
\fi
where $\text{erfc}(x)=\frac{1}{\pi}\int^\infty_x \exp\left(-t^2\right)dt$.
Note that although the bound in \eqref{UB2ofOP} is less insightful than 
that in \eqref{UB1ofOP}, it provides a tighter upper bound than 
that in \eqref{UB1ofOP} \cite{Simon00Exponential,Corazza_PoutUB}.

Fig. \ref{Fg:OutageProb} depicts the average secrecy outage probability, $\bar{P}_{\rm out}(w_k,\delta) $, with respect to $M$ and $N_d$  for different values of $w_k$ and $K$ when $\delta=0.1$. 
The results in Fig. \ref{Fig:OutageProb1} show that $\bar{P}_{\rm out}(w_k,\delta) $  decreases exponentially fast with $M$ as expected from \eqref{UBOnOP2}. It is also observed that a lower $\bar{P}_{\rm out}(w_k,\delta) $ is
achievable  as $w_k$ increases, i.e. a stronger PCA, which is due to the fact that a larger $w_k$ allows the proposed estimator $\hat{g}^e_k$ to have a smaller MSE as derived in \eqref{eq:MSE_gek}, and thereby we can reduce the occurrence of secrecy outage events. The secrecy outage probability also decreases as $N_d$ increases as shown in Fig. \ref{Fig:OutageProb3}, which is due to 
a better estimate of $g_k$ as observed in Fig. \ref{Fig:MSE2}. For a large $N_d$, $\bar{P}_{\rm out}(w_k,\delta) $ eventually converges to an certain value which corresponds to the average outage probability when the true value of $g_k$ is revealed to the estimator of $g^e_k$. 

\begin{figure*}[t]
	\centering
		\subfigure[Contour plot of $\bar{P}_{\rm out}(w_k,\delta) $ versus $M$ and $K$ when ${N_d = 1,000}$ and ${\delta = 0.1}$.]
		{
			\includegraphics[width=\SubFigWidth]{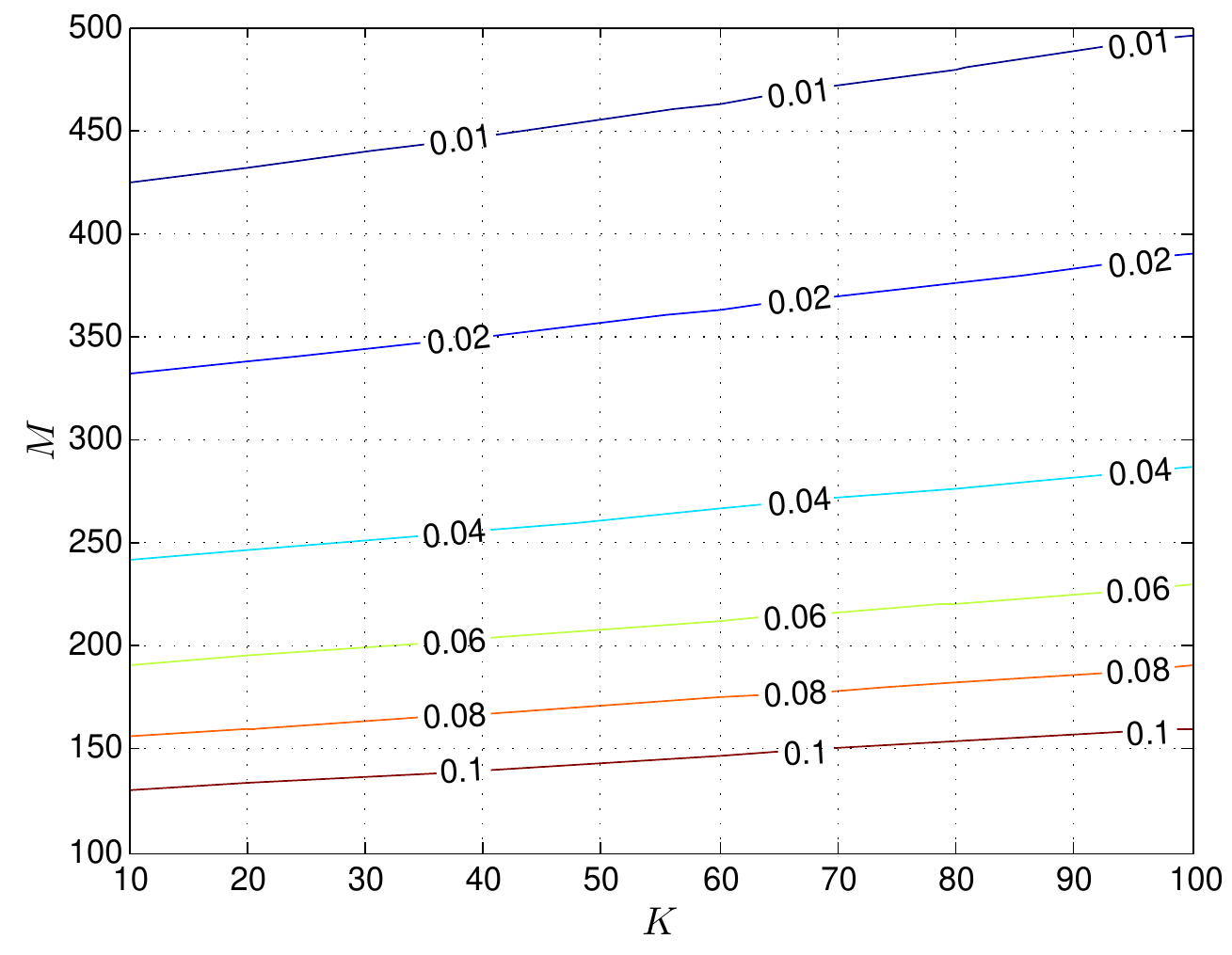}
			\label{Fig:Contour1}
		}
		\subfigure[Contour plot of $\bar{P}_{\rm out}(w_k,\delta) $ versus $M$ and $N_d$ when ${K=100}$ and ${\delta=0.1}$.]
		{
			\includegraphics[width=\SubFigWidth]{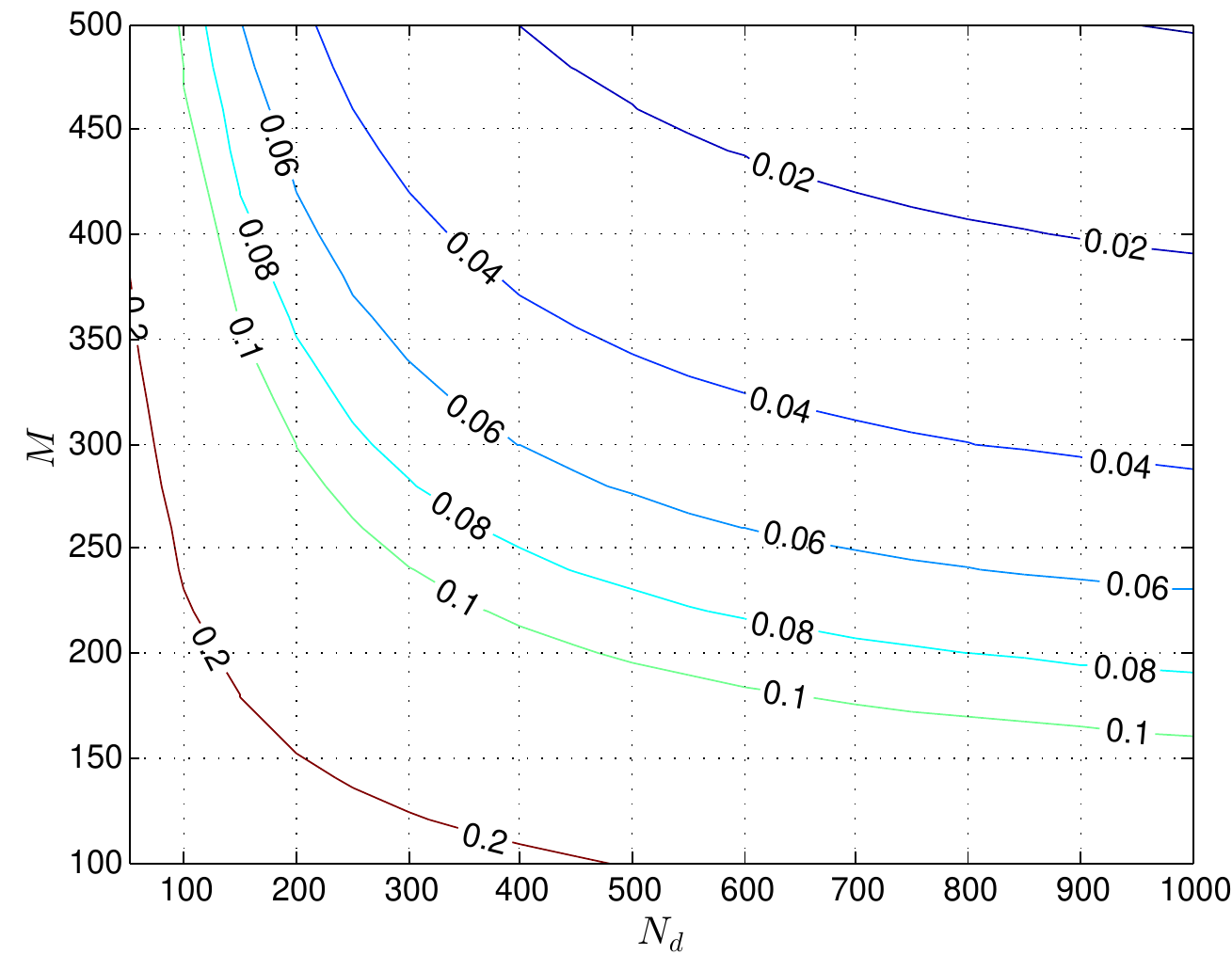}
			\label{Fig:Contour2}
		}
		\subfigure[Contour plot of $\bar{P}_{\rm out}(w_k,\delta) $ versus $M$ and $\delta$ when ${K=100}$ and $ {N_d = 1,000}$.]
		{
			\includegraphics[width=\SubFigWidth]{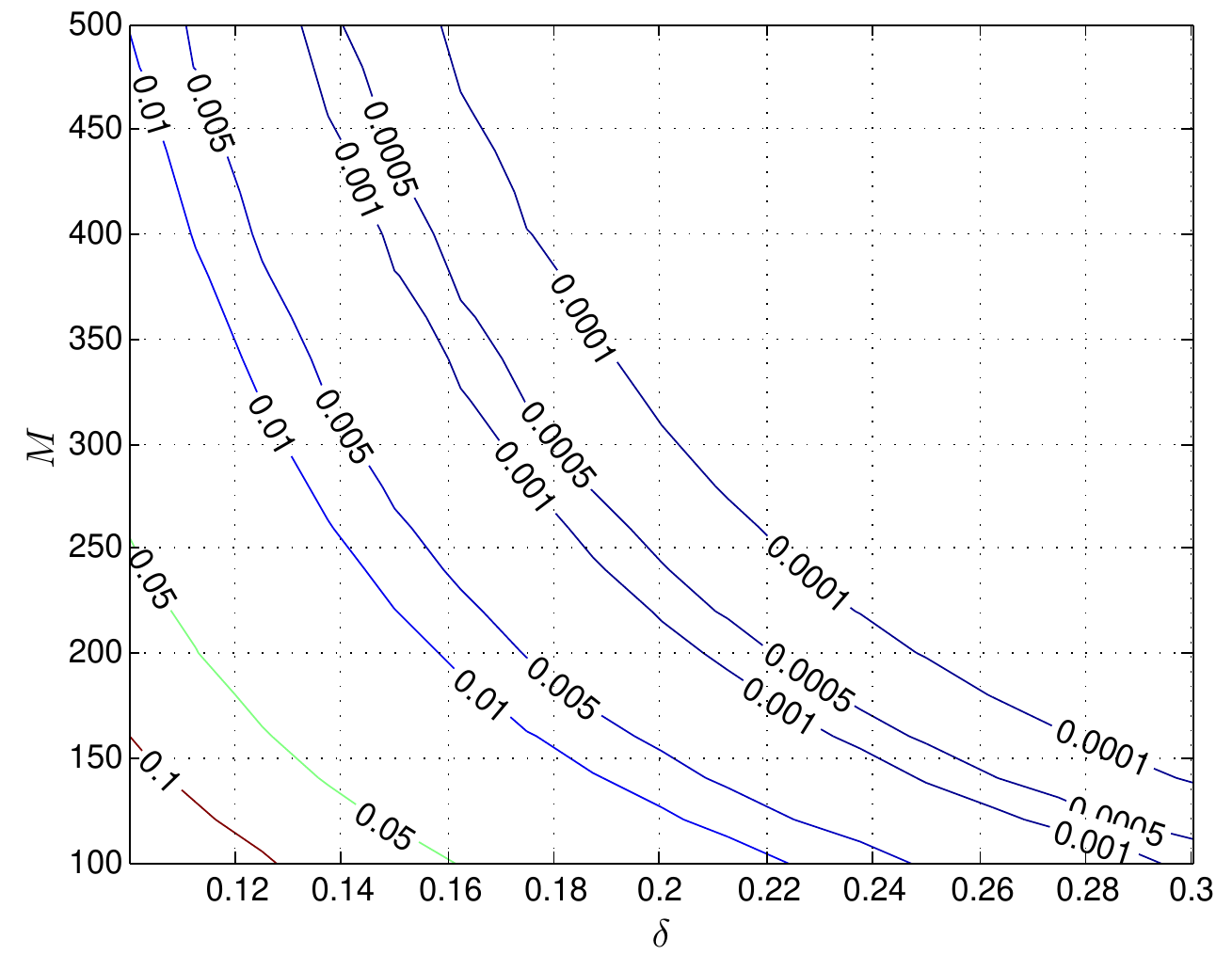}
			\label{Fig:Contour3}
		}
	\caption{Contour plot of $\bar{P}_{\rm out}(w_k,\delta) $ when $w_k^2 = -6 $ dB}
	\label{Fg:ContourPlots}
\end{figure*}
 
We now see trade-off relations among different parameters, $M$, $K$, $N_d$, and $\delta$, in the average secrecy outage probability by investigating the contour plots of $\bar{P}_{\rm out}(w_k,\delta) $ in Fig. \ref{Fg:ContourPlots} where $w_k^2$ is set to $-6$dB. The results provide a useful reference that enables system designers to select appropriate combinations of system parameters to meet various system requirements. The results in Fig. \ref{Fig:Contour1} show that the number of antennas $M$ must be almost linearly increased to compensate for the growing number of users $K$ to achieve a target average secrecy outage probability. The tradeoff between $M$ and $K$ is due to the results in \eqref{UBOnOP2} where the exponent is inversely proportional to $M/K$ for a large $N_d$. The average secrecy outage probability is also analyzed with respect to $M$ and $N_d$ in Fig. \ref{Fig:Contour2} where it is observed that the performance loss caused by employing a smaller size of the LAA can be compensated for by increasing $N_d$ to some extent and vice versa. However, as noticed in Fig. \ref{Fig:Contour2} the impact of $N_d$ on $\bar{P}_{\rm out}(w_k,\delta) $ is saturated fast with growing $N_d$. Thus, increasing $M$ is a more effective way to reduce $\bar{P}_{\rm out}(w_k,\delta) $ when $N_d$ is already large enough. Finally, we see the variations of average secrecy outage probability with respect to the secrecy margin, $\delta$ in Fig. \ref{Fig:Contour3} where it is observed that a small change of $\delta$ can exert a large influence on $\bar{P}_{\rm out}(w_k,\delta)$ since the average outage probability is decreases exponentially fast with respect to the square of the secrecy margin as shown in \eqref{UBOnOP2}. However, it should be noted that the smaller outage probability is achieved at the expense of the secret key length. We will discuss this issue  in detail in the next subsection.

\subsection{Average Length of Extracting Secret Key}\label{Average Secrecy Rate}
\begin{figure*}[t]
	\centering
		\subfigure[${R}_s(w_k, \delta)$ versus $M$.]
		{
			\includegraphics[width=\SubFigWidth]{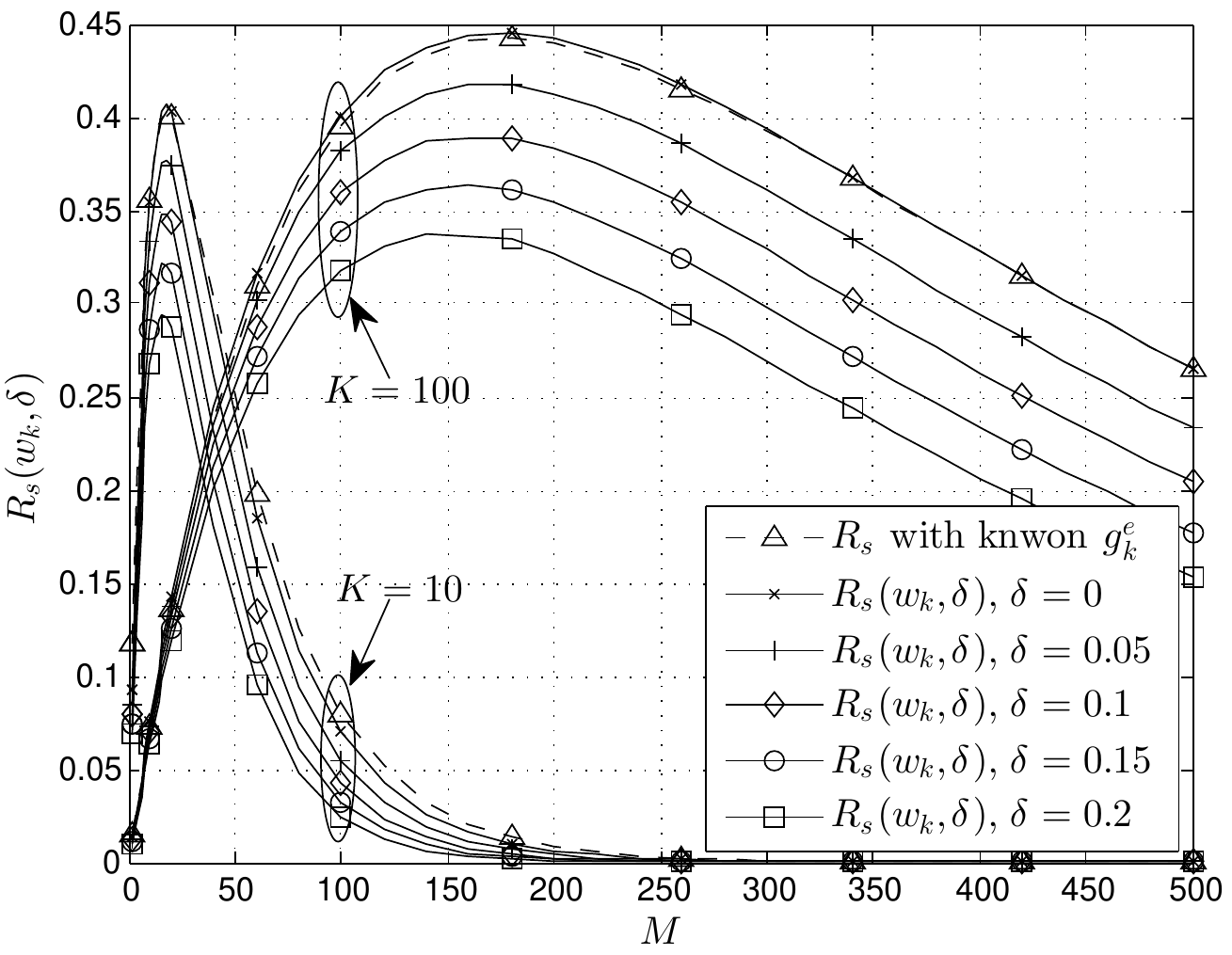}
			\label{Fig:AvgSKRs} 
		}\subfigure[$\bar{P}_{\rm out}(w_k,\delta) $ versus $M$ with different $\delta$.]
		{
			\includegraphics[width=\SubFigWidth]{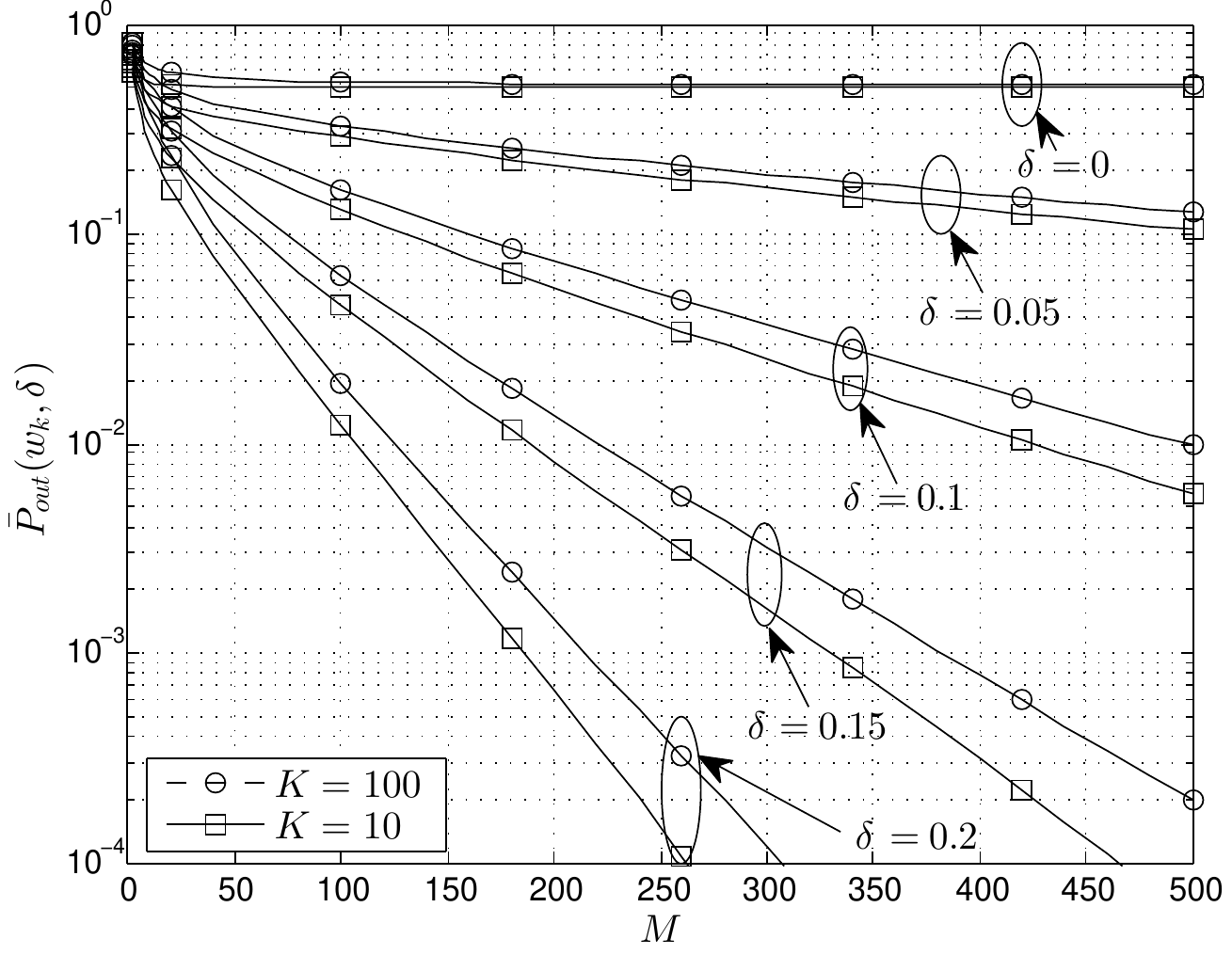}
			\label{Fig:OPsCPtoAvgSKRs} 
		}
	\caption{${R}_s(w_k, \delta)$ and $\bar{P}_{\rm out}(w_k,\delta) $ with respect to $M$ for $N_d = 1,000$ and $w_k^2 = -6$ dB. }
	\label{Fg:AvgSKRandOP} 
\end{figure*}

In this subsection, we evaluate the average length of secret key in \eqref{eq:Key3} with respect to different values of $M$ and $\delta$ when $N_d = 1,000$ and $w^2_k=-6$ dB. 
Note that $a_k$ and $b_k$ in \eqref{eq:Key3} become negligible for a sufficiently large $N_d$ \cite{Bloch08Wireless}. 
Hence, for simplicity, we only evaluate  ${R}_s(w_k, \delta) = \mathbb{E} \{ [I ( Q_k ; R_k ) - I ( Q_k ; \hat{R}^e_k )]^+ \}$, where the expectation is taken over the joint pdf of $f( \mathbf{r}_k , \mathbf{q}_k , \zeta_k ;w_k)$. We also evaluate the average secret key length when the eavesdropper's channel, $g^e_k$ is perfectly known to Bob, ${R}_s(w_k) = \mathbb{E} \{ [I ( Q_k ; R_k ) - I ( Q_k ; {R}^e_k )]^+ \}$ as a performance benchmark, which elucidates the performance loss due to the error in the estimation of $g^e_k$. 

The average secret key lengths ${R}_s(w_k, \delta)$ and average secrecy outage probabilities $\bar{P}_{\rm out}(w,\delta)$ are evaluated in Figs. \ref{Fig:AvgSKRs} and  \ref{Fig:OPsCPtoAvgSKRs}, respectively, with respect to $M$ for different values of $K$ and $\delta$. It is noticed that there exists a fundamental trade-off between ${R}_s(w_k, \delta)$ and $\bar{P}_{\rm out}(w_k,\delta)$. That is, if we increase $\delta$ to achieve a lower $\bar{P}_{\rm out}(w,\delta)$, we have accordingly a smaller ${R}_s(w_k,\delta)$. We can also observe that ${R}_s(w_k, \delta)$ with  $\delta = 0$ achieves almost the same performance of ${R}_s(w_k)$ in Fig. \ref{Fig:AvgSKRs}. This result implies that the proposed estimator $\hat{g}^e_k$ produces an estimate very close to its true value, $g^e_k$. However, in Fig. \ref{Fig:OPsCPtoAvgSKRs}, the corresponding $\bar{P}_{\rm out}(w_k,\delta) $  approaches almost 0.5. Note that the estimation result from the MMSE estimator for $g^e_k$ follows a Gaussian distribution \cite{Poor_1994, ChoiBook}, and thus $|\hat{g}^e_k|$ follows a Rician distribution. While the Rician distribution is not symmetric with respect to its true value $g^e_k$, it becomes symmetric as $M$ increases. Thus, for all sufficiently large $M$, $P_{\rm out} = \Pr(|\hat{g}^e_k| < |g^e_k|) = \Pr(|\hat{g}^e_k| \ge |g^e_k|) = 0.5$. However, for some small $M$ values, the asymmetry of the Rician distribution makes $P_{\rm out}$ larger than 0.5 as shown in Fig. \ref{Fig:OPsCPtoAvgSKRs}. However, the results in \eqref{UBOnOP} tell that $\bar{P}_{\rm out}(w,\delta)$ decreases exponentially fast with increasing $\delta$, and thus a small sacrifice of ${R}_s(w_k, \delta)$ is well paid off by a significant improvement of $\bar{P}_{\rm out}(w,\delta)$. 

In Fig. \ref{Fig:AvgSKRs}, it is also noticed that ${R}_s(w_k,\delta)$ has a peak after which ${R}_s(w_k,\delta)$ decreases with growing $M$. This happens due to the fact that both $I(Q_k;R_k)$ and $I(Q_k;R^e_k)$ can not exceed 1 bit per channel-use (bpcu) and are proportional to the size of the LAA as expected from \eqref{eq:SINR}. Thus, as $M$ grows, the ${R}_s(w_k,\delta)$ has a maximum value and later becomes diminished as both $I(Q_k;R_k)$ and $I(Q_k;R^e_k)$ approach 1 bpcu.
 
\section{Conclusions}\label{Sec:Con}
We studied an SKA protocol with LAA for a multi-user TDD 
system in the presence of multiple eavesdroppers attempting 
the PCA. By exploiting the complementary relation between the received signal 
strengths at the eavesdropper and its target user, 
an estimator was derived to measure the EDCG from 
the BS to the eavesdropper. From an estimated
EDCG, the amount of information leakage was quantified,
which was used to adaptively adjust 
the length of extract secret key.
Extensive performance evaluations have been carried out 
in both numerical and analytic ways. 
We showed that even in the case that the eavesdropper can manipulate
the LAA by the PCA, we can still take 
advantage of the LAA in the estimator and extract 
a certain length of secret keys with an arbitrary low 
secrecy outage probability. As future research directions, 
we will study a coordinated protocol to detect the PCA in 
a multi-cell scenario, where the reuse of 
the same training sequence across cells can be mis-identified as the PCA.
In addition, the research will further proceed to the cases when the channels between legitimate parities and eavesdropper are correlated, multiple antennas are employed in eavesdroppers, and the channel reciprocity does not hold due to the channel variations over time for practical considerations. 
\appendices
\section{}\label{App:A}
Consider 
two $M \times 1$ independent random vectors $\mathbf{x}=[x_1, \cdots, x_M]^T$ and $\mathbf{y}=[y_1, \cdots, y_M]^T$ that are zero-mean CSCG with covariance matrices $\sigma_x^2 \mathbf{I}_M$ and $\sigma_y^2 \mathbf{I}_M$, respectively. 
Then, the $m$-th component of $\mathbf{y}$ can be rewritten by $y_m = r_m e^{j \phi_m}$ for $m = 1, \cdots, M$ where $r_m \in [0,\infty)$ and $\phi_m \in [-\pi,\pi)$ follow Rayleigh and uniform distributions, respectively.
Then, it can be shown 
\begin{align*}
	t= \mathbf{x}^\dag \frac{\mathbf{y}}{\left\| \mathbf{y} \right\|}=\tilde{x}_1 \tilde{r}_1+\cdots +\tilde{x}_M \tilde{r}_M,
\end{align*}
where $\tilde{x}_m =x_m {e}^{j\phi_m}$ and $\tilde{r}_m=\frac{ r_m }{\sqrt{ r_1^2+\cdots + r_M^2}}$. Note that $\tilde{x}_m$ has the same distribution as $x_m$ due to the circularly-symmetric property. The pdf of $t$ is given by 
\ifCLASSOPTIONonecolumn 
\[
	 f (t)  = \int f(t | \tilde{\mathbf{r}}) f(\tilde{\mathbf{r}}) d\tilde{\mathbf{r}} 
	   \mathop = \limits^{(a)} \int \frac{1}{\pi\sigma^2_x \sum^{M}_{m=1} \tilde{r}_m^2 } \exp \left(-\frac{ t^2 }{\sigma^2_x \sum^{M}_{m=1} \tilde{r}_m^2 } \right) f(\tilde{\mathbf{r}}) d\tilde{\mathbf{r}} 
	  \mathop = \limits^{(b)} \frac{1}{\pi\sigma^2_x} \exp \left( -\frac{ t^2 }{\sigma^2_x} \right),
\]
\else 
\begin{align*}
	 f (t) 
	 &= \int f(t | \tilde{\mathbf{r}}) f(\tilde{\mathbf{r}}) d\tilde{\mathbf{r}} \nonumber \\
	 &  \mathop = \limits^{(a)} \int \frac{1}{\pi\sigma^2_x \sum^{M}_{m=1} \tilde{r}_m^2 } \exp \left(-\frac{ t^2 }{\sigma^2_x \sum^{M}_{m=1} \tilde{r}_m^2 } \right) f(\tilde{\mathbf{r}}) d\tilde{\mathbf{r}} \nonumber \\
	 & \mathop = \limits^{(b)} \frac{1}{\pi\sigma^2_x} \exp \left( -\frac{ t^2 }{\sigma^2_x} \right),
\end{align*}
\fi
where $(a)$ is due to the fact that $t$ can be seen as the summation of independent complex Gaussian random variables $\{ \tilde{x}_m \}$ for given $\tilde{\mathbf{r}} = [\tilde{r}_1, \cdots, \tilde{r}_M]^T $, and $(b)$ is due to the fact that $ \sum^{M}_{m=1} \tilde{r}_m^2 = 1$.

\section{}\label{App:SINR}
Since we normalize the average power of $\mathbf{q}_k$ to one, we have $\text{SINR}_k = \frac{ \mathbb{E} \left[ |g_k |^2   \right] }{ \sigma_{n_k}^2}$ that is determined by the distribution of $g_k$. 
Based on the orthogonality principle of the MMSE estimation \cite{Poor_1994, ChoiBook}, we can rewrite $g_k$ as follows: 
\begin{align} \label{eq:SINR1}
	g_k 
	  = \frac{\mathbf{h}_k^\dag \mathbf{a}_k}{\sqrt{M}}  
	  \mathop = \limits^{(a)} \frac{( \hat{\mathbf{h}}_k + \mathbf{e}_k )^\dag }{\sqrt{M}} \frac{\hat{\mathbf{h}}_k }{||\hat{\mathbf{h}_k}||} ,
\end{align} 
where $(a)$ is from $\mathbf{a}_k  = \frac{\mathbf{y}_k}{\zeta_k} = \frac{\hat{\mathbf{h}}_k}{|| \hat{\mathbf{h}}_k||} $, and $\mathbf{e}_k$ is the estimation error of the MMSE estimation. 
Note that we have  $\hat{\mathbf{h}}_k \sim \mathcal{CN} (\mathbf{0}_M, \frac{c_k}{1+ (1+w_k^2)c_k} \mathbf{I}_M)$ and $ \mathbf{e}_k \sim \mathcal{CN} (\mathbf{0}_M, \frac{1 + w_k^2 c_k}{1+ (1+w_k^2)c_k} \mathbf{I}_M)$ from the MMSE property \cite{Poor_1994, ChoiBook}. 
Then, $g_k$ in \eqref{eq:SINR1} can be rewritten by  
\begin{align*} 
	g_k 
	  = \frac{ 1 }{\sqrt{M}} \left( || \hat{\mathbf{h}}_k || +  \mathbf{e}_k ^\dag \frac{ \hat{\mathbf{h}_k} }{||\hat{\mathbf{h}_k}||} \right) 
	  \mathop = \limits^{(a)} \frac{ 1 }{\sqrt{M}} \left( || \hat{\mathbf{h}}_k || +  e_k \right) ,
\end{align*} 
where $(a)$ is from Appendix \ref{App:A}, and $e_k \sim \mathcal{CN} (0, \frac{1+ w_k^2 c_k}{1+ (1+w_k^2) c_k})$.
We notice that due to the orthogonality principle of the MMSE estimation, we have $\mathbb{E} [|| \hat{\mathbf{h}}_k ||   e_k] =0 $. 
Thus, we have $\mathbb{E}[ |g_k |^2] = \frac{1}{M} \mathbb{E} [ | || \hat{\mathbf{h}}_k || +  e_k |^2 ]  = \frac{1}{M} \mathbb{E} [ || \hat{\mathbf{h}}_k ||^2 ] + \frac{1}{M} \mathbb{E} [ | e_k |^2 ]$. 
Since both $ \hat{\mathbf{h}}_k $ and $  e_k$ follow the complex Gaussian distribution, we have $\mathbb{E} [ || \hat{\mathbf{h}}_k ||^2 ] = \frac{ c_k}{2\{1+ (1+w_k^2)c_k\}} \mathbb{E}[\mathcal{X}^2_{2M}]$ and $ \mathbb{E} [ | e_k |^2 ] = \frac{1 + w_k^2 c_k}{2\{ 1+ (1+w_k^2)c_k\} } \mathbb{E}[\mathcal{X}^2_{2}]$, where $\mathcal{X}^2_m$ is a chi-square random variable with $m$ degrees of freedom whose first order moment is $\mathbb{E}[\mathcal{X}^2_{m}]=m$. 
Thus, we have $\mathbb{E} [ |g_k |^2 ] = \frac{ c_k  }{1+ (1+w_k^2)c_k} + \frac{1 + w_k^2 c_k}{ M\{1+ (1+w_k^2)c_k \}}	= \frac{M c_k + w_k^2 c_k +1 }{M\{ 1 + (1+w_k^2) c_k \} } $.
Then, from \eqref{EQ:sn}, we finally have 
\ifCLASSOPTIONonecolumn
\begin{align*}
	\text{SINR}_k 
	= \frac{M c_k + w_k^2 c_k +1 }{( 1 + (1+w_k^2) c_k ) ( K-1 + \frac{1}{p_d \beta_k})} .
\end{align*} 
\else
\begin{align*}
	\text{SINR}_k 
	= \frac{M c_k + w_k^2 c_k +1 }{(1 + (1+w_k^2) c_k ) ( K-1 + \frac{1}{p_d \beta_k})} .
\end{align*} 
\fi
We omit the derivation of $\text{SINR}^e_k$ since we can derive $\text{SINR}^e_k$ in the same manner.

\section{}\label{App:Limit_EDCG}
Let us first find the asymptotic behavior of  $f (g_k | \zeta_k ; w )$.
In Appendix \ref{App:pdfs}, we have $f (g_k | \zeta_k ; w ) \sim \mathcal{CN} (\mu_{g_k}, \sigma^2_{g_k})$. 
It is obvious that $\lim_{M\to \infty} \sigma^2_{g_{k}} =0$, while the limiting value of $\mu_{g_k}$ is given by  
\ifCLASSOPTIONonecolumn 
\begin{align} \label{eq:Limit_EDCG}
	\lim_{M\to \infty} \mu_{g_k} 
	&= \lim_{M \to \infty} \frac{\zeta_k }{\sqrt{M}} \left\{ \frac{\sqrt{c_k}}{1 + (1 + w^2_k) c_k} \right\} 
	 = \lim_{M \to \infty} \frac{ 1 }{ \sqrt{M} } \sqrt{  \sum\nolimits_{m = 1}^{M} \left| y_{km} \right|^2} \left\{ \frac{\sqrt{c_k}}{1 + (1 + w^2_k) c_k} \right\} \nonumber\\
	& \mathop = \limits^{(a)} \sqrt{ \frac{c_k}{1 + (1 + w^2_k) c_k}},
\end{align}
\else
\begin{align} \label{eq:Limit_EDCG}
	\lim_{M\to \infty} \mu_{g_k} 
	&= \lim_{M \to \infty} \frac{\zeta_k }{\sqrt{M}} \left\{ \frac{\sqrt{c_k}}{1 + (1 + w^2_k) c_k} \right\} \nonumber\\
	& = \lim_{M \to \infty} \frac{ 1 }{ \sqrt{M} } \sqrt{  \sum\nolimits_{m = 1}^{M} \left| y_{km} \right|^2} \left\{ \frac{\sqrt{c_k}}{1 + (1 + w^2_k) c_k} \right\} \nonumber\\
	& \mathop = \limits^{(a)} \sqrt{ \frac{c_k}{1 + (1 + w^2_k) c_k}},
\end{align}
\fi
where $y_{k m}$ is the $m$-th component of  $\mathbf{y}_k $, and $(a)$ is from the law of large numbers that, as $M$ goes to infinity, the sample variance of random variable converges to its true variance.
Therefore, for a given $\zeta_k$, $g_k \to \mu_{g_k}$ in probability  as ${M \to \infty}$. 
Then, the asymptotic behavior of $f (g_k ; w_k)$ for a large $M$ is given by  
\ifCLASSOPTIONonecolumn 
\begin{align*}
	\lim_{M \to \infty} f (g_k ; w_k)
	& = \lim_{M \to \infty} \int f (g_k | \zeta_k ; w_k ) f(\zeta_k ; w_k) d \zeta_k 
	  \mathop = \limits^{(a)} \int \lim_{M \to \infty}  f (g_k | \zeta_k ; w_k )  f(\zeta_k ; w_k) d \zeta_k  \\
	& \mathop = \limits^{(b)} \sqrt{ \frac{c_k}{1+(1 + w^2_k) c_k}} \text{ in probability},
\end{align*}
\else
\begin{align*}
	& \lim_{M \to \infty} f (g_k ; w_k) \nonumber\\
	& = \lim_{M \to \infty} \int f (g_k | \zeta_k ; w_k ) f(\zeta_k ; w_k) d \zeta_k  \\
	& \mathop = \limits^{(a)} \int \lim_{M \to \infty}  f (g_k | \zeta_k ; w_k )  f(\zeta_k ; w_k) d \zeta_k  \\
	& \mathop = \limits^{(b)} \sqrt{ \frac{c_k}{1+(1 + w^2_k) c_k}} \text{ in probability},
\end{align*}
\fi
where $(a)$ is due to the Lebesgue dominated convergence theorem, and $(b)$ is from \eqref{eq:Limit_EDCG}.

\section{}\label{App:pdfs}
In this Appendix, we derive various pdfs used in this paper\footnote{Throughout this paper, we do not include all details of the derivation if it involves simple calculations.}.
\subsection{pdfs of $f (g_k | \zeta_k ; w_k )$ and $f (g^e_k | \zeta_k ; w_k)$}
We first derive the pdf of $f (g_k | \zeta_k ; w_k )$.
To obtain $f (g_k | \zeta_k ; w_k )$, we first derive $f(\mathbf{h}_k | \mathbf{y}_k ;w_k)$ using the Baye's theorem, which is given by 
\begin{align} 
	f(\mathbf{h}_k | \mathbf{y}_k ; w_k) = \frac{ f(\mathbf{y}_k| \mathbf{h}_k ; w_k) f( \mathbf{h}_k)}{ f(\mathbf{y}_k ; w_k)},
	\label{eq:Bay}
\end{align} 
where, from \eqref{eq:y_k}, we have  $f(\mathbf{y}_k| \mathbf{h}_k ; w_k)  \sim {\mathcal{CN} (\sqrt{c_k} \mathbf{h}_k , (1+w_k^2 c_k) \mathbf{I}_M)}$, $f( \mathbf{h}_k )  \sim \mathcal{CN} ( \mathbf{0}_M ,  \mathbf{I}_M )$, and $ f(\mathbf{y}_k; w_k)  \sim \mathcal{CN} ( \mathbf{0}_M, \left( 1+ c_k + w_k^2 c_k \right) \mathbf{I}_M )$.
After some manipulations, we can derive the conditional pdf of $f(\mathbf{h}_k | \mathbf{y}_k ; w_k)$ from \eqref{eq:Bay} as follows:
\begin{align*}
	f(\mathbf{h}_k | \mathbf{y}_k ; w_k) \sim \mathcal{CN} \left( \frac{\sqrt{c_k}}{1 + c_k + w_k^2 c_k} \mathbf{y}_k,  \frac{ 1 + w_k^2 c_k }{ 1 + c_k + w_k^2 c_k }\mathbf{I}_{M} \right).
\end{align*}
Note that $g_k = \mathbf{h}_k^T \mathbf{a}_k / \sqrt{M}$ is the sum of scaled Gaussian random variables for  a given $\mathbf{y}_k$ since $\mathbf{a}_k$ becomes constant for a given $\mathbf{y}_k$. Then, we have $f(g_k | \mathbf{y}_k ; w_k) \sim \mathcal{CN} ( \mu_{g_k} , \sigma^2_{g_k} )$, where
\begin{align*}
	\mu_{g_k} = \frac{\sqrt{c_k}}{1 + c_k + w^2_k c_k}  \frac{ \norm {\mathbf{y}_k } }{\sqrt{M}} 
	\text{ and }
	\sigma^2_{g_k} = \frac{ 1 + w^2_k c_k }{(1 +c_k + w^2_k c_k ) M}.
\end{align*}
Finally, $f \left( g_k | \norm {\mathbf{y}_k } ; w_k \right) $ is derived from $ f \left( g_k | {\mathbf{y}_k } ; w_k\right)$ as follows:
\ifCLASSOPTIONonecolumn 
\begin{align*}
	f \left( g_k | \mathbf{y}_k ; w_k\right) 
		= f \left( g_k | \mathbf{y}_k, \norm {\mathbf{y}_k } ; w_k\right)
	 	\mathop =\limits^{(a)} f \left( g_k | \norm {\mathbf{y}_k } ;w_k \right),
\end{align*}
\else
\begin{align*}
	f \left( g_k | \mathbf{y}_k ; w_k\right) 
		& = f \left( g_k | \mathbf{y}_k, \norm {\mathbf{y}_k } ; w_k\right) \nonumber \\
	 	& \mathop =\limits^{(a)} f \left( g_k | \norm {\mathbf{y}_k } ;w_k \right),
\end{align*}
\fi
where $(a)$ is due to the fact that $\mu_{g_k}$ and $\sigma^2_{g_k}$ are independent of $\mathbf{y}_k$ for a given $\norm {\mathbf{y}_k }$. 
In the same manner, we have $f (g^e_k | \zeta_k ; w_k) \sim \mathcal{CN} (\mu_{g^e_k} , \sigma^2_{g^e_k})$, where
$
	\mu_{g^e_k} = \frac{w_k\sqrt{c_k}}{(1+ c_k + w^2_k c_k )M} \frac{\zeta_k}{\sqrt{M}} 
$ and $
	\sigma^2_{g^e_k} = \frac{1 + c_k }{ ( 1+c_k + w^2_k c_k ) M}.
$
\subsection{pdf of $f ( g^e_k | g_k , \zeta_k ; w_k ) $}
We first derive the joint pdf of $f (\mathbf{h}^e_k , \mathbf{h}_k , \mathbf{y}_k ; w_k )$. 
From the Baye's theorem, we have 
\begin{align*}
	f (\mathbf{h}^e_k , \mathbf{h}_k , \mathbf{y}_k  ; w_k ) 
	&= f ( \mathbf{h}_k , \mathbf{y}_k  |\mathbf{h}^e_k  ; w_k ) f (\mathbf{h}^e_k  ) ,
\end{align*}
where 
\ifCLASSOPTIONonecolumn 
$
f(\mathbf{h}_k, \mathbf{y}_k  | \mathbf{h}^e_k  ; w_k ) 
  \sim \mathcal{CN}\left( \bigl[\begin{smallmatrix} \mathbf{0}_M & w_k\sqrt{c_k} \mathbf{h}^e_k \end{smallmatrix} \bigr]  , \bigl[\begin{smallmatrix} \mathbf{I}_{M} & \sqrt{c_k} \mathbf{I}_{M} \\ \sqrt{c_k} \mathbf{I}_{M} & (1+ c_k) \mathbf{I}_{M} \end{smallmatrix} \bigr] \right)
$
and $f(\mathbf{h}^e_k)  \sim \mathcal{CN}(\mathbf{0}_{M} , \mathbf{I}_{M})$.
\else
\begin{align*}
f(& \mathbf{h}_k, \mathbf{y}_k  | \mathbf{h}^e_k  ; w_k ) \\
  &\sim \mathcal{CN}\left( \bigl[\begin{smallmatrix} \mathbf{0}_M & w_k\sqrt{c_k} \mathbf{h}^e_k \end{smallmatrix} \bigr]  , \bigl[\begin{smallmatrix} \mathbf{I}_{M} & \sqrt{c_k} \mathbf{I}_{M} \\ \sqrt{c_k} \mathbf{I}_{M} & (1+ c_k) \mathbf{I}_{M} \end{smallmatrix} \bigr] \right)
\end{align*}
and $f(\mathbf{h}^e_k)  \sim \mathcal{CN}(\mathbf{0}_{M} , \mathbf{I}_{M})$.
\fi

Then, the conditional distribution of $\mathbf{h}^e_k$ given $\mathbf{h}_k $ and $ \mathbf{y}_k$, $f (\mathbf{h}^e_k | \mathbf{h}_k , \mathbf{y}_k ; w_k ) $, becomes also the multivariate normal distribution that can be directly obtained from $f (\mathbf{h}^e_k , \mathbf{h}_k , \mathbf{y}_k ; w_k )$ \cite{Eaton}.
After some manipulations, we have
$f (\mathbf{h}^e_k | \mathbf{h}_k , \mathbf{y}_k ; w_k ) \sim \mathcal{CN} (\boldsymbol{\mu}, \boldsymbol{\Sigma})$, where $\boldsymbol{\mu} = \frac{w_k \sqrt{c_k} }{1+ w_k^2 c_k } \left( \mathbf{y}_k - \sqrt{c_k} \mathbf{h}_k \right) $ and $\boldsymbol{\Sigma} = \frac{1  }{  1+ w^2_k c_k  }\mathbf{I}_M$.
Then, since $g^e_k = \frac{(\mathbf{h}^e_k)^T \mathbf{a}_k}{\sqrt{M}}$, we can derive $f (g^e_k | \mathbf{h}_k , \mathbf{y}_k ; w_k )$ from the transformation of the random vector as follows:   
\ifCLASSOPTIONonecolumn
\begin{align*}
	f (g^e_k | \mathbf{h}_k , \mathbf{y}_k ; w_k ) \sim \mathcal{CN} \left( \frac{w_k\sqrt{c_k}}{1+ w^2_k c_k } \left( \frac{\zeta_k}{\sqrt{M}} - \sqrt{c_k} g_k \right) , \frac{1}{( 1+ w^2_k c_k) M} \right).
\end{align*}
\else
\begin{align*}
	f(& g^e_k| \mathbf{h}_k , \mathbf{y}_k ; w_k ) \nonumber \\
	&\sim \mathcal{CN} \left( \frac{w_k\sqrt{c_k}}{1+ w^2_k c_k } \left( \frac{\zeta_k}{\sqrt{M}} - \sqrt{c_k} g_k \right) , \frac{1}{( 1+ w^2_k c_k) M} \right).
\end{align*}
\fi
Note that, by substituting \eqref{eq:estimate_h_k} into \eqref{eq:MF_precoder}, we have $\mathbf{a}_k  = \frac{\mathbf{y}_k}{\zeta_k}$, which provides  
$f (g^e_k | \mathbf{h}_k , \mathbf{y}_k ; w_k ) = f \left( g^e_k | \mathbf{h}_k  , \mathbf{a}_k,  \mathbf{y}_k, \zeta_k ; w_k \right) = f \left( g^e_k | \mathbf{h}_k  , \mathbf{a}_k,  g_k, \mathbf{y}_k, \zeta_k ; w_k \right)$. 
Then, we finally have $f (g^e_k | \mathbf{h}_k , \mathbf{y}_k ; w_k ) = f (g^e_k | g_k , \zeta_k ; w_k )$ since $f (g^e_k | \mathbf{h}_k , \mathbf{y}_k ; w_k )$ is characterized by $ g_k $ and $ \zeta_k$.
\section{}\label{App:MMSE_gk}
The MMSE estimator for $g_k$ is  given by 
$\hat{g}_{k} 
	= \mathbb{E} \left[ g_k | \mathbf{r}_k , \mathbf{q}_k , \zeta_k ; w_k \right] 
	= \int  g_{k} f\left(g_{k} | \mathbf{r}_k , \mathbf{q}_k , \zeta_k ; w_k \right) d g_k . 
$
Thus, the mean value of $f\left(g_{k} | \mathbf{r}_k , \mathbf{q}_k , \zeta_k ; w_k \right)$, denoted by $\mu_{\hat{g}_k}$, is the MMSE estimator for $g_k$. 
The distribution of $f\left(g_{k} | \mathbf{r}_k , \mathbf{q}_k , \zeta_k ; w_k \right)$ can be derived as follow: 
\ifCLASSOPTIONonecolumn 
\begin{align} \label{eq:AppE1}
	f\left(g_{k} | \mathbf{r}_k , \mathbf{q}_k , \zeta_k ; w_k \right)
	= \frac{f\left(g_{k} , \mathbf{r}_k | \mathbf{q}_k , \zeta_k ; w_k \right)}{f\left(\mathbf{r}_k | \mathbf{q}_k , \zeta_k ; w_k \right)}
	= \frac{ f\left(\mathbf{r}_k | g_k,  \mathbf{q}_k ;w_k \right) f\left( g_k | \zeta_k ; w_k \right)  }{\int f\left(\mathbf{r}_k | g_k,  \mathbf{q}_k  ;w_k \right) f\left( g_k | \zeta_k ; w_k \right) d g_k}  .
\end{align}
\else
\begin{align} \label{eq:AppE1}
	&f\left(g_{k} | \mathbf{r}_k , \mathbf{q}_k , \zeta_k ; w_k \right) \nonumber\\
	&= \frac{f\left(g_{k} , \mathbf{r}_k | \mathbf{q}_k , \zeta_k ; w_k \right)}{f\left(\mathbf{r}_k | \mathbf{q}_k , \zeta_k ; w_k \right)} \nonumber\\
	&= \frac{ f\left(\mathbf{r}_k | g_k,  \mathbf{q}_k  ; w_k \right) f\left( g_k | \zeta_k ; w_k \right)  }{\int f\left(\mathbf{r}_k | g_k,  \mathbf{q}_k  ; w_k \right) f\left( g_k | \zeta_k ; w_k \right) d g_k}  .
\end{align}
\fi
The pdf of $f\left( g_k | \zeta_k ; w_k \right)$ in \eqref{eq:AppE1} is derived in Appendix \ref{App:pdfs}, while the pdf of $f\left(\mathbf{r}_k | g_k,  \mathbf{q}_k  ; w_k \right)$ can be  obtained from \eqref{eq:r_Bob2}, which is given by 
$	f\left(\mathbf{r}_k | g_k,  \mathbf{q}_k  ; w_k \right) \sim \mathcal{CN} \left( g_k  \mathbf{q}_k , \sigma_{n_k}^2 \mathbf{I}_{N_d}\right).
$ 
Then, after some  manipulations, we finally have the pdf of $f\left(g_{k} | \mathbf{r}_k , \mathbf{q}_k , \zeta_k ; w_k \right) \sim \mathcal{CN} (\mu_{\hat{g}_k} , \sigma^2_{\hat{g}_k})$, where 
\begin{align*}
	\mu_{\hat{g}_k}  =  \frac{\mathbf{r}_k^\dag \mathbf{q}_k + \sigma^2_{n_k} \mu_{g_k}/ \sigma^2_{g_k} }{\mathbf{q}_k^\dag \mathbf{q}_k + \sigma^2_{n_k}/ \sigma^2_{g_k} } 
	\text{ and }
	\sigma^2_{\hat{g}_k} =  \frac{\sigma^2_{n_k} \sigma^2_{g_k} }{\mathbf{q}_k^\dag \mathbf{q}_k \sigma^2_{g_k}+ \sigma^2_{n_k} } .
\end{align*}
Thus, the MMSE estimator for $g_k$ is given by $\hat{g}_{k}  = \frac{\mathbf{r}_k^\dag \mathbf{q}_k + \sigma^2_{n_k} \mu_{g_k}/ \sigma^2_{g_k} }{\mathbf{q}_k^\dag \mathbf{q}_k + \sigma^2_{n_k}/ \sigma^2_{g_k} } $.
In the same manner, we can derive  the pdf of $f\left(g^e_k | \mathbf{r}_k , \mathbf{q}_k , \zeta_k ; w_k \right) \sim \mathcal{CN} (\mu_{\hat{g}^e_k} , \sigma^2_{\hat{g}^e_k})$, where 
\begin{align*}
	\mu_{\hat{g}^e_k} & =  \frac{w_k c_k}{1+ w^2_k c_k} \left(\frac{\zeta_k}{\sqrt{c_k M}}-\frac{\mathbf{r}_k^\dag \mathbf{q}_k + \sigma^2_{n_k} \mu_{g_k}/ \sigma^2_{g_k} }{\mathbf{q}_k^\dag \mathbf{q}_k + \sigma^2_{n_k}/ \sigma^2_{g_k} }  \right) \text{ and }\\
	\sigma^2_{\hat{g}^e_k} & = \frac{1}{(1+w^2_k c_k)M} + \left(\frac{w_k c_k}{1+w^2_k c_k}\right)^2 \frac{\sigma^2_{g_k}\sigma^2_{n_k}}{\sigma^2_{g_k} \mathbf{q}^\dag_k \mathbf{q}_k+\sigma^2_{n_k}} .
\end{align*}

\bibliographystyle{IEEEtran}
\bibliography{[Ref].bib}
\end{document}